\documentclass[10pt]{article}

\usepackage{amssymb}
\usepackage{amsmath}
\usepackage{algpseudocode}
\usepackage{array}
\usepackage{amsfonts}
\usepackage{amssymb}
\usepackage{amsopn}
\usepackage{graphicx}
\usepackage{authblk}
\usepackage{mathtools}
\usepackage[mathlines]{lineno}
\usepackage{graphics,epsfig}
\DeclareGraphicsExtensions{.pdf}

\usepackage{color}

\usepackage{url,here}
\usepackage[backref,colorlinks=true]{hyperref}
\usepackage{multirow}
\usepackage{algorithm}

\allowdisplaybreaks

\begin{document}
\newenvironment {proof}{{\noindent\bf Proof.}}{\hfill $\Box$ \medskip}

\newtheorem{theorem}{Theorem}[section]
\newtheorem{lemma}[theorem]{Lemma}
\newtheorem{condition}[theorem]{Condition}
\newtheorem{proposition}[theorem]{Proposition}
\newtheorem{remark}[theorem]{Remark}
\newtheorem{definition}[theorem]{Definition}
\newtheorem{hypothesis}[theorem]{Hypothesis}
\newtheorem{corollary}[theorem]{Corollary}
\newtheorem{example}[theorem]{Example}
\newtheorem{descript}[theorem]{Description}
\newtheorem{assumption}[theorem]{Assumption}

\def\longrightharpoonup{\relbar\joinrel\rightharpoonup}
\def\longleftharpoondown{\leftharpoondown\joinrel\relbar}

\def\longrightleftharpoons{
  \mathop{
    \vcenter{
      \hbox{
      \ooalign{
        \raise1pt\hbox{$\longrightharpoonup\joinrel$}\crcr
	  \lower1pt\hbox{$\longleftharpoondown\joinrel$}
	  }
      }
    }
  }
}

\newcommand{\rates}[2]{\displaystyle
\mathrel{\longrightleftharpoons^{#1\mathstrut}_{#2}}}
\newcommand{\br}[1]{\langle #1\rangle}

\def\P{\mathbb{P}}
\def\R{\mathbb{R}}
\def\C{\mathbb{C}}
\def\E{\mathbb{E}}
\def\N{\mathbb{N}}
\def\Z{\mathbb{Z}}

\renewcommand {\theequation}{\arabic{section}.\arabic{equation}}
\def \non{{\nonumber}}
\def \hat{\widehat}
\def \tilde{\widetilde}
\def \bar{\overline}

\def\ind{{\mathchoice {\rm 1\mskip-4mu l} {\rm 1\mskip-4mu l}
{\rm 1\mskip-4.5mu l} {\rm 1\mskip-5mu l}}}

\title{\bf Dynamic disorder in simple enzymatic reactions induces stochastic amplification of substrate}
\author[1]{Ankit Gupta}
\author[1,2]{Andreas Milias-Argeitis}
\author[1]{Mustafa Khammash\thanks{mustafa.khammash@bsse.ethz.ch}}
\affil[1]{Department of Biosystems Science and Engineering, ETH Zurich, Mattenstrasse 26, 4058 Basel, Switzerland.}
\affil[2]{Groningen Biomolecular Sciences and Biotechnology Institute, University of Groningen, Nijenborgh 4, 9747 AG Groningen, the Netherlands.}

\date{}
\maketitle

\begin{abstract}
A growing amount of evidence over the last two decades points to the fact that many enzymes exhibit fluctuations in their catalytic activity, which are associated with conformational changes on a broad range of timescales. The experimental study of this phenomenon, termed dynamic disorder, has become possible thanks to advances in single-molecule enzymology measurement techniques, through which the catalytic activity of individual enzyme molecules can be tracked in time. The biological role and importance of these fluctuations in a system with a small number of enzymes such as a living cell, have only recently started being explored.

In this work, we examine a simple stochastic reaction system consisting of an inflowing substrate and an enzyme with a randomly fluctuating catalytic reaction rate that converts the substrate into an outflowing product. To describe analytically the effect of rate fluctuations on the average substrate abundance at steady-state, we derive an explicit formula that connects the relative speed of enzymatic fluctuations with the mean substrate level. Under fairly general modeling assumptions, we demonstrate that the relative speed of rate fluctuations can have a dramatic effect on the mean substrate, and lead to large positive deviations from predictions based on the assumption of deterministic enzyme activity. Our results also establish an interesting connection between the amplification effect and the mixing properties of the Markov process describing the enzymatic activity fluctuations, which can be used to easily predict the fluctuation speed above which such deviations become negligible. As the techniques of single-molecule enzymology continuously evolve, it may soon be possible to study the stochastic phenomena due to enzymatic activity fluctuations within living cells. Our work can be used to formulate experimentally testable hypotheses regarding the nature and magnitude of these fluctuations, as well as their phenotypic consequences.
\end{abstract} 

 

\setcounter{equation}{0}
\section{Introduction}
First made almost two decades ago, observations of enzymatic turnovers for single enzyme molecules have allowed scientists to probe enzyme behavior beyond the regime of high-copy numbers and ensemble averages \cite{Lu98}. Thanks to advances brought by experimental techniques such as single-molecule fluorescence spectroscopy \cite{Enderlein11,Lu14}, the field of single-molecule enzymology developed rapidly in the subsequent years. The key observation made possible by single-molecule assays is that the catalytic rates of single enzyme molecules often display very large dynamic fluctuations over timescales much longer than the typical reaction cycle times, most likely driven by slow (spontaneous or induced) transitions in conformation \cite{Lu98,Lu05,English06,Xie02}.

Around the time when the first single enzyme molecules were observed in action, the mathematical theory of \emph{dynamic disorder} was introduced by Zwanzig \cite{Zwanzig90}, motivated by several observations of different physico-chemical processes with a seemingly common underlying cause that boiled down to random fluctuations of key process properties. The phenomenon of dynamic disorder refers to fluctuations in enzymatic reaction rates that occur at a timescale that is either slower or comparable to the reaction timescale \cite{Dan07}. These fluctuations are often caused by \emph{slow} transitions in the conformational state of enzymes. The most simple example of dynamic disorder first considered by Zwanzig \cite{Zwanzig90} involved a so-called ``rate process controlled by passage through a fluctuating bottleneck'' \cite{Zwanzig92}. In the language of chemical kinetics, it describes the removal of a substrate, $S$, from a system at rate $\gamma(t)S(t)$, where the time-varying rate $\gamma(t)$ is a (typically Markovian) stochastic process (see Figure \ref{fig:model}). As the speed of $\gamma(t)$ fluctuations tends to zero, the reaction rate $\gamma(t)$ becomes a random variable which does not change with time, and we transition to the regime of \emph{static} disorder \cite{Zwanzig90}. On the other hand, as the speed of $\gamma(t)$ fluctuations tends to infinity, the dynamic disorder \emph{vanishes} on the timescale of substrate kinetics, and we recover the classical case where the reaction rate $\gamma(t)$ becomes a deterministic constant. Our goal in this paper is to investigate the effects of $\gamma(t)$ fluctuations between these two extremities, where most realistic systems are likely to lie.

Thanks to the mathematical theory of dynamic disorder, stochastically fluctuating enzyme activities can be understood and studied within a consistent mathematical framework that can also generate testable experimental predictions \cite{Schenter99,Kou05,Reichman06,Dan07,English06,Lu98}. Already in \cite{Zwanzig90} it was observed that dynamically disordered systems can give rise to macroscopic observations that differ from those expected in the absence of disorder. In subsequent years, a large body of theoretical and computational work has examined various alternative enzymatic reaction schemes, mostly focusing on the enzyme dynamics itself, e.g. on the autocorrelation of fluctuations and the distribution of waiting times between turnover events \cite{min2005}. On the other hand, dynamic disorder has been observed for several biologically relevant enzymes \cite{Tan11}, suggesting that it is ubiquitous in the cellular context. Early work \cite{English06} had already noted that enzymatic fluctuations could play an important biological role in a system containing only a small number of enzyme molecules, as often happens within a living cell, and the recent \emph{in vivo} observation of fluctuating enzymatic activity confirms this claim \cite{Iversen14}.

Besides studying the intrinsic mathematical properties of dynamically disordered enzymes, it would be also highly instructive and relevant for biology to examine the \emph{consequences} of dynamic disorder on substrate statistics (a first example of such a study is given in \cite{Iversen14}). Experimental work in this area is still done \emph{in vitro} using constant and large substrate pools. Here, on the contrary, we provide a mathematical treatment of how dynamic disorder alters the substrate mean abundance in the presence of substrate \emph{inflow}, a condition closer to biological reality that has also been considered in \cite{Grima14,Stefanini05}. To this end, we analytically examine a highly simplified stochastic system with a randomly fluctuating catalytic reaction rate and describe the effect of rate fluctuations on the average substrate abundance. Under fairly general conditions, we demonstrate that the relative speed of rate fluctuations can have a dramatic effect on the mean substrate, and lead to large positive deviations from predictions based on the assumption of deterministic enzyme activity. Using a Markovian model for enzyme kinetics, we mathematically characterize this effect by deriving an explicit formula for the steady-state substrate-mean as a function of the relative speed of enzymatic fluctuations. From this formula we show that for any finite speed-value, the steady-state substrate-mean is \emph{sandwiched} between the two values obtained in the static and the deterministic regimes. Furthermore, we demonstrate that the mapping between the relative speed of enzyme kinetics and the substrate-mean at steady-state can be well-approximated by a convex, monotonically decreasing function whose key shape parameter depends on the ``mixing strength" of the Markov process describing enzyme kinetics. This mixing strength can be measured by computing an appropriate  \emph{Dirichlet form} \cite{Saloff-Coste} of the Markov process. Even though we consider a highly simplified situation, our analysis can serve as a guide in the case of more realistic, but analytically intractable enzymatic reaction schemes. Our results \emph{only} depend on the enzymatic fluctuations, but they do not depend on the fluctuations caused by the low abundance of substate molecules (see \cite{Elowitz}), although we account for these fluctuations by modelling the substrate kinetics as a jump Markov chain. Indeed the results we present remain unchanged even if we discard these fluctuations and describe the substrate kinetics as an ordinary differential equation (ODE) with a fluctuating rate constant $\gamma(t)$.

Enzymatic fluctuations can also arise from sources other than dynamic disorder. The abundance or the availability of enzyme molecules may also fluctuate due to gene expression noise \cite{Elowitz,Oyarzun}, and their chances of finding substrate molecules can be diffusion-limited \cite{DiffLim}. In this work, we do not distinguish between the various sources of fluctuations and model the aggregate enzymatic activity by a Markovian stochastic process.

The biological significance of our findings is manifold since enzymatic interactions are ubiquitous is cell biology and the effects of enzymatic noise in metabolic networks have only recently started to be explored \cite{Oyarzun}. Using the relative speed of enzymatic fluctuations as parameter, our results provide a clear way to determine if the deterministic approximation is a faithful representation of reality. Our results can shed light on the timescale disparities that exist between enzyme and substrate kinetics. In particular, we see that enzyme kinetics needs to be ``fast'' in order to avoid any undesirable amplification of the mean substrate abundance due to inevitable variations in the enzymatic states. On the other hand, one can envisage situations where it would be beneficial for enzymes to be ``slow'' so that their fluctuations amplify a weak signal and enable its detection by the intracellular machinery (see Section \ref{example:twostate}). Such a signal-detection mechanism was the main motivation behind \emph{stochastic focusing}, a sensitivity amplification phenomenon introduced in \cite{Paulsson00}. We illustrate our results on the reaction scheme of \cite{Paulsson00} in Section \ref{example:sf}, where we characterize how the substrate-mean changes with the speed of the enzyme abundance dynamics. Note that in situations where enzyme kinetics is ``slow'', undesirable amplification effects can be eliminated by feedback mechanisms \cite{Milias15}. In this context, our results can help in postulating the presence of feedback loops using experimental data. We discuss the biological importance of our results in greater detail in Section \ref{sec:disc}.

It is interesting to note that some of the expressions we derive are related to those obtained in the analysis of various physico-chemical quantum dynamical systems coupled to a randomly fluctuating environment. This theory dates back to the original work of Kubo and Anderson \cite{Anderson54,Kubo54} and, more recently, has been generalized to arbitrary quantum systems described by the Liouvile-von Neumann equation with Markovian and non-Markovian parametric noise \cite{goychuk2004quantum,goychuk2005quantum,goychuk2005rate}. The example system of Ref. \cite{goychuk2005rate} can be interpreted as a substrate decaying with a stochastically fluctuating rate. While similar to it, the system we consider here includes the inflow of substrate, which results in a non-zero steady-state and requires a different mathematical treatment.

\section{Results}
\subsection{The model}\label{sec:model}
We consider a system into which a substrate, ${\bf S}$, enters at a constant rate $k_{ \textnormal{in} }$ and is degraded or (equivalently) converted into a product ${\bf P}$ that in turn leaves the system. The rate of substrate outflow depends on the activity state or abundance level of an enzyme ${\bf E}$. In turn, the catalytic activity of ${\bf E}$, denoted by $(\gamma(t) )_{t \geq 0 }$, is assumed to fluctuate in time $t$ according to a continuous-time Markov chain (CTMC) over a finite state-space $\Gamma = \{ \gamma_1,\dots,\gamma_n \}$. Here each $\gamma_i$ is a positive constant denoting the degradation rate constant at the $i$-th enzymatic state or abundance level. Due to fluctuations in the catalytic activity of ${\bf E}$, the degradation rate of substrate ${\bf S}$ will also fluctuate in time according to a stochastic process $( k_{d,S}(t) )_{t \geq 0}$ whose value at time $t$ is given by $k_{d,S}(t)= \gamma(t)S(t)$, where $S(t)$ is the molecular count or \emph{concentration} of the substrate. This model is summarised in Figure \ref{fig:model}.
\begin{figure}[!h]
\begin{center}
    \includegraphics[width=0.4\textwidth]{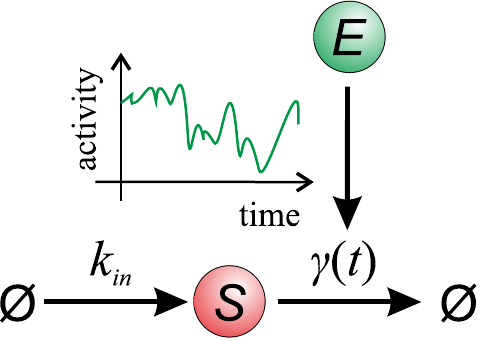}
\end{center}
    \caption{Schematic of the modeled system.}
    \label{fig:model}
\end{figure}

As mentioned before, the degradation reaction ${\bf S} \stackrel{ \gamma(t) }{\longrightarrow} {\bf \emptyset}$ can also be viewed as a conversion reaction ${\bf S} \stackrel{ \gamma(t) }{\longrightarrow} {\bf P}$ which is catalyzed by the enzyme. Generally this catalytic step proceeds through the \emph{reversible} formation of an intermediate complex ${\bf S.E}$ which is formed when an enzyme molecule binds to a substrate molecule. In other words, the single reaction $ {\bf S} \stackrel{ \gamma(t) }{\longrightarrow} {\bf P}$ is an abstraction for the following three reactions:
\begin{linenomath*}
\begin{align*}
{\bf S} + {\bf E} \longrightleftharpoons {\bf S.E}  \longrightarrow {\bf P} + {\bf E}.
\end{align*}
\end{linenomath*}
If the binding/unbinding rates of ${\bf S}$ and ${\bf E}$ molecules is much higher than the rate of the conversion reaction, then we can apply the quasi-stationary assumption to conclude that the model in Figure \ref{fig:model} is a good approximation to the catalytic conversion dynamics (for more details see the Supplementary Material in \cite{Milias15}).

To describe the CTMC $(\gamma(t) )_{t \geq 0 }$, we need to specify its $n \times n$ transition rate matrix $Q=[q_{ij}]$ (see \cite{Norris98}). For any distinct $i,j \in \{1,2,\dots,n\}$, $q_{ij} \geq 0$ denotes the rate at which the process leaves state $\gamma_i$ and enters state $\gamma_j$. The diagonal entries of $Q$ are given by $q_{ii} = -\sum_{j \neq i}q_{ij}$. From now on we assume that the rate matrix $Q$ is irreducible\footnote{A matrix $Q$ is called \emph{irreducible} if there does not exist a permutation matrix $P$ such that the matrix $P Q P^{-1}$ is block upper-triangular.} which implies that there exists a unique stationary distribution $\pi = (\pi_1,\dots,\pi_n) \in \R^n_+$ satisfying
\begin{linenomath*}
\begin{align*}
Q {\bf 1} = {\bf 0},  \qquad  \pi^T Q  = {\bf 0}^T  \qquad \textnormal{and} \qquad \pi^T {\bf 1} = 1,
\end{align*}
\end{linenomath*}
where ${\bf 0}$ and ${\bf 1}$ denote the $n \times 1$ vectors of all zeroes and ones respectively. Since the state-space is finite and the transition rate matrix $Q$ is irreducible, the CTMC $( \gamma (t) )_{t  \geq 0 }$ is \emph{ergodic} which means that the probability distribution of $\gamma(t)$ converges to the stationary distribution $\pi$ as $t \to \infty$. As we are interested in the steady-state limit, without loss of generality we can assume that the initial state $\gamma(0)$ is distributed according to $\pi$, i.e. $\P( \gamma(0) = \gamma_i )$ for each $i=1,\dots,n$. This ensures that the process $(\gamma(t) )_{t \geq 0 }$ is a \emph{stationary} stochastic process whose various statistical properties do not depend on time\footnote{For a more rigorous definition of a stationary stochastic process see the Supplementary Material.}. In particular its mean $\E( \gamma(t) )$ is equal to
\begin{linenomath*}
\begin{align}
\label{stationarymean:gamma}
\E( \gamma(t) ) = \E_\pi(\gamma)=  \sum_{i=1}^n \gamma_i \pi_i \qquad \textnormal{for all} \quad t \geq 0,
\end{align}
\end{linenomath*}
where $\gamma$ is a $\Gamma$-valued random variable with probability distribution $\pi$ and $\E_\pi(\cdot)$ denotes the expectation w.r.t. this distribution.

From now on, we regard $( \gamma(t) )_{t \geq 0}$ as the \emph{baseline} process which corresponds to enzymatic dynamics at the natural timescale. In order to study the substrate behavior, we need to view enzymatic dynamics at the timescale of substrate kinetics. For this we define a family of processes $( \gamma_c(t) )_{t  \geq 0 }$ parameterised by the ``relative speed" parameter $c$ as follows:
\begin{linenomath*}
\begin{align}
\label{defngammac}
\gamma_c(t) = \gamma(c t) \qquad \textnormal{for all} \quad t \geq 0.
\end{align}
\end{linenomath*}
Note that one time-unit of process $( \gamma_c(t) )_{t  \geq 0 }$ corresponds to $c$ time-units of process $( \gamma(t) )_{t \geq 0}$. In this sense, the parameter $c$ sets the speed of the fluctuation dynamics for the enzyme \emph{relative} to the speed of the substrate kinetics. Like $( \gamma(t) )_{t \geq 0}$, the process $( \gamma_c(t) )_{t \geq 0}$ is also a CTMC over state-space $\Gamma = \{ \gamma_1,\dots,\gamma_n \}$ with transition rate matrix $Q_c= cQ$ and initial distribution $\pi$. Since $( \gamma(t) )_{t \geq 0}$ is stationary, this process is also stationary with the same mean given by $\E_\pi(\gamma) = \E( \gamma_c(t) ) $ for all times $t \geq 0$. Replacing $( \gamma(t) )_{t \geq 0}$ by $( \gamma_c(t) )_{t  \geq 0 }$ in the model depicted in Figure \ref{fig:model}, we will study how the steady-state mean of substrate abundance depends on the fluctuation speed $c$.

Given a sample path of the enzyme dynamics $( \gamma_c(t) )_{t  \geq 0 }$ with relative speed $c$, we regard the dynamics of substrate molecular counts as a jump Markov chain $( S_c(t))_{t \geq 0}$ over the set of nonnegative integers $\N_0= \{0,1,2,\dots\}$. This Markov chain can be written in the random time change representation \cite{EK} as
\begin{linenomath*}
\begin{align}
\label{sub_rtc}
S_c(t) = S_c(0) + Y_1(  k_{ \textnormal{in} } t  ) - Y_2\left( \int_0^t \gamma_c(u) S_c(u)du \right),
\end{align}
\end{linenomath*}
where $Y_1$ and $Y_2$ are independent, unit rate Poisson processes. From this representation it is immediate that the substrate-production rate is constant ($k_{\textnormal{in} }$) in time, but the substrate-degradation rate is time-varying and it is equal to $\gamma_c(t) S_c(t)$ at time $t$. Here the Poisson processes $Y_1$ and $Y_2$ capture the intermittency in the firing of production and degradation reactions. This intermittency becomes unimportant if the substrate is present in high copy-numbers \cite{KurtzLLn1} and in this case one can regard $( S_c(t))_{t \geq 0}$ as the dynamics of subtrate \emph{concentration}\footnote{The concentration of any species is its copy-number divided by the system volume.}, specified by the following ODE
\begin{linenomath*}
\begin{align}
\label{sub_ode}
\frac{ d S_c(t)}{d t} = k_{ \textnormal{in} } -  \gamma_c(t) S_c(t).
\end{align}
\end{linenomath*}
Note that even if the intermittency in production/degradation reactions is ignored and $( S_c(t))_{t \geq 0}$ is described by the ODE \eqref{sub_ode}, the process $( S_c(t))_{t \geq 0}$ is still \emph{stochastic} because it is driven by the stochastic process $( \gamma_c(t) )_{t  \geq 0 }$ that represents enzymatic fluctuations.

Let $m_c(t) = \E(S_c(t))$ for each $t \geq 0$. We shall soon see that $m_c(t)$ does not depend on whether we use representation \eqref{sub_rtc} or \eqref{sub_ode} for the substrate dynamics $( S_c(t) )_{t \geq 0}$. Our goal in this paper is to understand the role of fluctuations in the catalytic activity of enzyme ${\bf E}$ in determining the steady-state value of the mean
\begin{linenomath*}
\begin{align}
\label{steadystatemean}
m_\textnormal{eq}(c) = \lim_{t \to \infty} m_c(t).
\end{align}
\end{linenomath*}
In particular, we study how this steady-state mean $m_\textnormal{eq}(c)$ depends on the relative fluctuation speed $c$ and the variability in degradation rates $\gamma_1,\dots,\gamma_n$ at various enzymatic activity levels.

\subsection{Expressions for $m_{ \textnormal{eq} }(c)$: The general case } \label{sec:mainanalysis}

We can approximately find $m_\textnormal{eq}(c) $ by estimating $m_c(t) $ for a very large $t$, using simulations of the whole system. However this naive approach is highly unsatisfactory because these simulations can be computationally expensive and the approximation error incurred by replacing the steady-state mean by a finite-time mean is generally difficult to quantify. Moreover this approach does not provide us with an explicit formula for $m_\textnormal{eq}(c) $ that can enable us to study its dependence on the relative speed parameter $c$. In light of these difficulties, we look for alternative ways to compute $m_\textnormal{eq}(c)$. In this section we assume that enzymatic kinetics is given by a general stationary stochastic process with an arbitrary state-space $\Gamma \subset (0,\infty)$, and so we do not rely on the CTMC structure mentioned in Section \ref{sec:model}. We specialise the results of this section to the CTMC case in Section \ref{sec:mainanalysisctmc}.

Using representation \eqref{sub_rtc} or \eqref{sub_ode} we can show that $m_c(t) =  \E(S_c(t))$ is given by the following formula
\begin{linenomath*}
\begin{align}
\label{formulamt}
m_c(t)=    \E \left( S_c(0) e^{ -  \int_0^t \gamma_c(s)ds } \right)  +  k_{  \textnormal{in} } \int_{0}^t   \E\left( e^{ -  \int_s^{t} \gamma_c(u)du }  \right)  ds.
\end{align}
\end{linenomath*}
From the stationarity of the process $( \gamma_c(t) )_{t \geq 0}$ we can conclude that
\begin{linenomath*}
\begin{align*}
 \int_{0}^t   \E \left( e^{ -  \int_s^t \gamma_c(u)du }  \right)  ds  =  \int_{0}^t   \E  \left( e^{ -  \int_0^{t-s} \gamma_c(u)du }  \right)  ds =  \int_{0}^t   \E  \left( e^{ -  \int_0^{s} \gamma_c(u)du }  \right)ds.
\end{align*}
\end{linenomath*}
Substituting this in \eqref{formulamt} and letting $t \to \infty$, we obtain our first formula for $m_\textnormal{eq}(c)$, which is,
\begin{linenomath*}
\begin{align}
\label{mceqform1}
m_{ \textnormal{eq} }(c) = \lim_{t \to \infty}m_c(t)=   k_{ \textnormal{in} } \int_{0}^\infty   \E \left( e^{ -  \int_0^{s} \gamma_c(u)du }    \right) ds.
\end{align}
\end{linenomath*}
From \eqref{defngammac} we obtain
\begin{linenomath*}
\begin{align}
\label{quantity1}
\int_0^{s} \gamma_c(u)du = s \left( \frac{1}{cs} \int_0^{cs} \gamma(u)du \right).
\end{align}
\end{linenomath*}
Since the process $( \gamma(t) )_{t \geq 0}$ is stationary, from Theorem 10.6 in \cite{Kal} we know that as $c \to \infty$, the quantity \eqref{quantity1} converges a.s. to $s \E_\pi( \gamma )$ (recall \eqref{stationarymean:gamma}). As a consequence $ \E \left( e^{ -  \int_0^{s} \gamma_c(u)du }    \right)  \to e^{ -s \E_\pi( \gamma ) }$ and hence we get
\begin{linenomath*}
\begin{align}
\label{mdeterministic}
\lim_{c \to \infty} m_{ \textnormal{eq} }(c) =  k_{ \textnormal{in} } \int_{0}^\infty  e^{ -s \E_\pi( \gamma ) } ds = \frac{ k_{ \textnormal{in} } }{ \E_\pi( \gamma ) } := m^{ \textnormal{(det)} }_{ \textnormal{eq} }.
\end{align}
\end{linenomath*}
This shows that as the relative speed $c$ of enzymatic fluctuations approaches $\infty$, these fluctuations become \emph{equilibrated} at the timescale of substrate kinetics, and so they do not affect the mean substrate level. In other words, from the point of view of the substrate, the enzyme kinetics is so fast that it is as if the enzyme state is constant at the equilibrium level $\E_\pi( \gamma )$. This corresponds to the classical case where there is no dynamic disorder in the enzyme activity and so this activity is well-approximated by a deterministic rate constant for the substrate degradation reaction. As the mapping $x \mapsto e^{-x}$ is convex, Jensen's inequality tells us that
\begin{linenomath*}
\begin{align*}
 \E \left( e^{ -  \int_0^{s} \gamma_c(u)du }    \right) \geq  e^{  - \int_0^s \E ( \gamma_c(u) )du  } = e^{  -  s \E_\pi( \gamma ) },
\end{align*}
\end{linenomath*}
where the last relation follows from the fact that $( \gamma_c(t) )_{t \geq 0}$ is a stationary process with mean $  \E( \gamma_c(t) ) = \E_\pi( \gamma ) $ for all times $t \geq 0$. Substituting this in \eqref{mceqform1} we see that for any $c \geq 0$
\begin{linenomath*}
\begin{align}
\label{positivedeviations}
m_{ \textnormal{eq} }(c) \geq  k_{ \textnormal{in} } \int_{0}^\infty  e^{ -s\E_\pi( \gamma )} ds = \frac{ k_{ \textnormal{in} } }{   \E_\pi( \gamma )} = m^{ \textnormal{(det)} }_{ \textnormal{eq} }.
\end{align}
\end{linenomath*}
Therefore for a finite relative speed $c$, enzymatic fluctuations \emph{always amplify the mean substrate abundance}, in comparison to the classical deterministic case. The natural question that now arises is - how large should speed $c$ be in order for the deterministic approximation to be acceptable within a certain tolerance level $\epsilon$? We address this question in Section \ref{sec:approxformula}.

Let us now consider the situation where the relative speed parameter $c \to 0$ and so at the timescale of substrate kinetics, the enzyme dynamics $( \gamma_c(t) )_{t \geq 0} $ approaches a {\bf static} process, i.e. $\gamma_c(t) = \gamma(0)$ for all $t \geq 0$. This case corresponds to the situation where the enzyme kinetics is \emph{very slow} in comparison to the substrate kinetics. Hence from the point of view of the substrate, the kinetics of the enzyme is almost \emph{fixed}. In this regime, we can replace $\gamma_c(u)$ by $\gamma(0)$ in \eqref{mceqform1} to obtain
\begin{linenomath*}
\begin{align}
\label{mstatic}
\lim_{c \to 0} m_{ \textnormal{eq} }(c) =  k_{ \textnormal{in} }\E\left(   \int_{0}^\infty  e^{ -s  \gamma(0) } ds \right)=  k_{ \textnormal{in} } \E \left( \frac{1}{\gamma(0)} \right) = k_{ \textnormal{in} } \E_\pi\left( \frac{1}{\gamma} \right):= m^{ \textnormal{(static)} }_{ \textnormal{eq} },
\end{align}
\end{linenomath*}
where we have used the fact that $\gamma(0)$ has probability distribution $\pi$ to write $\E( 1/\gamma(0) )$ as $\E_{\pi}(1/\gamma)$. Observe that $m^{ \textnormal{(static)} }_{ \textnormal{eq} } \geq m^{ \textnormal{(det)} }_{ \textnormal{eq} }$, which can be readily seen by letting $c \to 0$ in \eqref{positivedeviations} or by directly using Jensen's inequality on the convex map $f(x) = 1/x$ (see Figure \ref{fig:gammadistortion}). The two extremal cases $c \to 0$ and $c \to \infty$ serve as a guide to the behavior of realistic systems with an intermediate value of $c$. In particular we can expect that for such intermediate $c$-values, the steady-state substrate mean will lie somewhere between $m^{ \textnormal{(det)} }_{ \textnormal{eq} }$ and  $m^{ \textnormal{(static)} }_{ \textnormal{eq} }$ . This is precisely what happens as we shall soon see. We will also discuss how the precise value of $m_\textnormal{eq}(c)$ can be computed or estimated from any Markovian model of enzymatic fluctuations.

\begin{figure}[!h]
\begin{center}
    \includegraphics[width=0.4\textwidth]{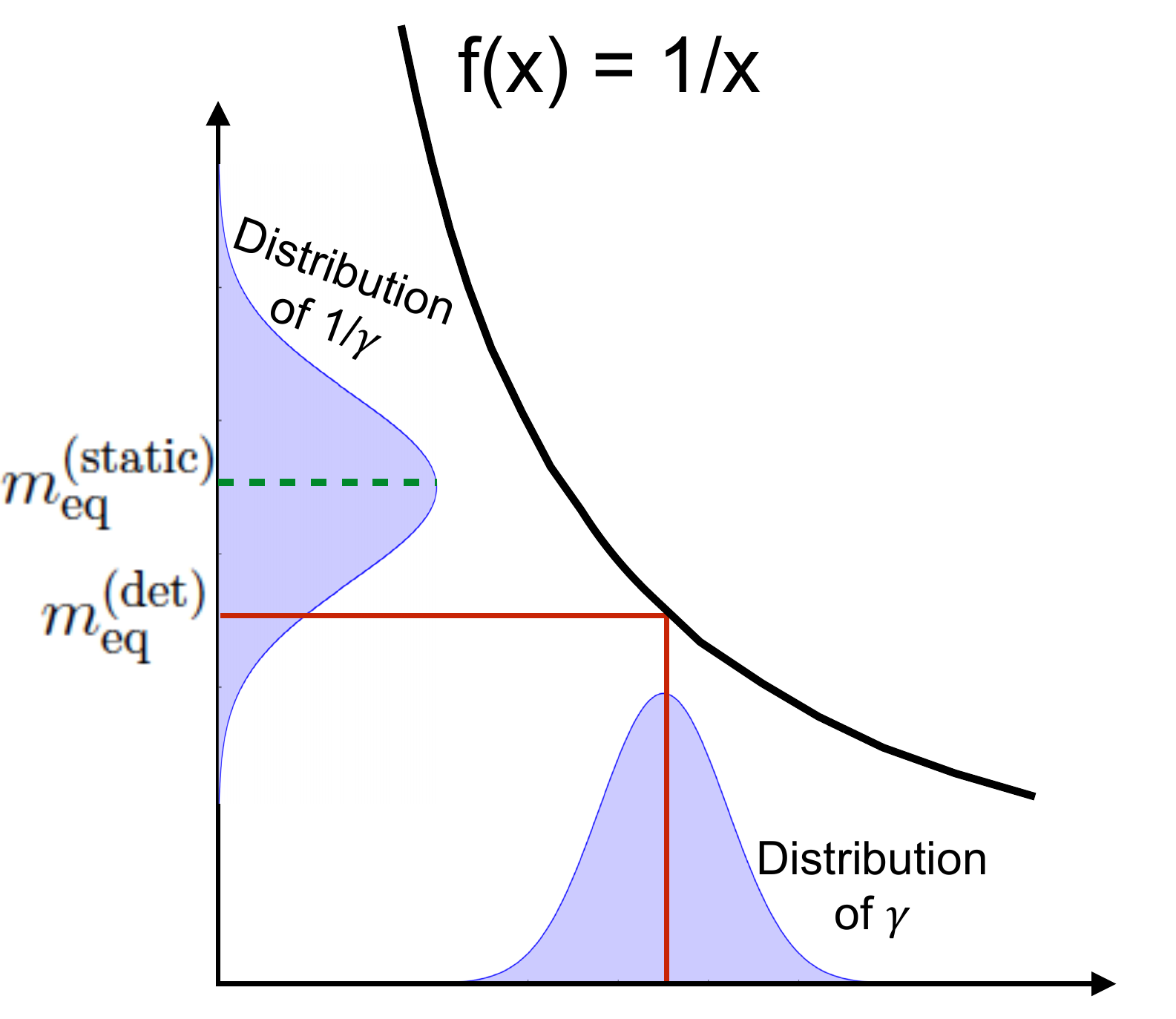}
\end{center}
    \caption{This figure illustrates how the convex map $f(x) = 1/x$ distorts the distribution of $\gamma$. Assuming $ k_{ \textnormal{in} }= 1$ here, the static mean $m^{ \textnormal{(static)} }_{ \textnormal{eq} } = \E_\pi( f(\gamma) ) = \E_{\pi}( 1/\gamma )$ is always greater than the deterministic mean $m^{ \textnormal{(det)} }_{ \textnormal{eq} } = f( \E_\pi( \gamma ) ) = 1/\E_{\pi}( \gamma )$.}
    \label{fig:gammadistortion}
\end{figure}

Until now, the conclusions we have drawn regarding $m_{ \textnormal{eq} }(c) $ rely on the formula \eqref{mceqform1} that holds for any real-valued stationary stochastic process $( \gamma_c(t) )_{t \geq 0}$ as long as its states are positive and bounded away from $0$, i.e. the state-space $\Gamma$ satisfies
\begin{linenomath*}
\begin{align}
\label{gamma_inf}
\inf\{x : x \in \Gamma \}  \geq \epsilon
\end{align}
\end{linenomath*}
for some $\epsilon >0$. This makes this formula very general but it is difficult to work with, because it involves an indefinite integral which is generally analytically intractable as the mapping $s \mapsto  \E \left( e^{ -  \int_0^{s} \gamma_c(u)du }    \right)$ does not have an explicit form. We remedy this problem in the next section by specialising this formula to the case where $( \gamma_c(t) )_{t \geq 0}$ is a finite state-space CTMC as described in Section \ref{sec:model}. Before we come to that, we provide a numerical recipe for statistically estimating $m_{ \textnormal{eq} }(c) $ without the need for evaluating the indefinite integral. This scheme is based on the assumption that we can efficiently generate sample-paths of the stationary process $( \gamma_c(t) )_{t \geq 0}$ (see \cite{GP,Asmussen}).

Define a random variable $\tau_c$ by
\begin{linenomath*}
\begin{align}
\label{defn:tau}
\tau_c = \inf\left\{ t \geq 0 :  \int_0^t \gamma_c(s)ds = - \ln u \right\},
\end{align}
\end{linenomath*}
where $u$ is an independent random variable with the uniform distribution on $[0,1]$. To sample $\tau_c$ we can adopt the following strategy. We first sample $u$ from the uniform distribution on $[0,1]$, draw an initial condition $\gamma_c(0)$ from $\pi$, and then simulate the sample path $( \gamma_c(t) )_{t \geq 0}$, keeping track of the integral $ \int_0^t \gamma_c(s)ds $. We take $\tau_c$ to be first time $t$ when this integral hits the value $(-\ln u)$. From the samples of the random variable $\tau_c$, we can estimate its expectation $\E( \tau_c )$ which gives us an estimate for $m_{ \textnormal{eq} }(c)$ because it can be shown that
\begin{linenomath*}
\begin{align}
\label{eqnetauc}
m_{ \textnormal{eq} } (c) = k_{ \textnormal{in} } \E( \tau_c ).
\end{align}
\end{linenomath*}
Note that the estimator for $m_{ \textnormal{eq} } (c)$ based on formula \eqref{eqnetauc} will be unbiased but it will suffer from statistical error due to a finite sample size. However this error can be estimated and managed far more easily than the error one would incur by approximating the steady-state mean $m_{ \textnormal{eq} } (c)$ by a finite-time mean $m_c(t)$ (recall \eqref{steadystatemean}). This makes this formula \eqref{eqnetauc} useful in practice (see Example \ref{example:sf}).

The results from this section are collected in our next proposition which is proved in the Supplementary Material.
\begin{proposition}
\label{prop_general}
Suppose $ ( \gamma(t) )_{t \geq 0 }$ is a real-valued stationary stochastic process with stationary distribution $\pi$ and state-space $\Gamma$ satisfying \eqref{gamma_inf}. Let $ ( \gamma_c(t) )_{t \geq 0}$ be the speed $c$ version of this process given by \eqref{defngammac} and define the substrate dynamics $( S_c(t) )_{t \geq 0 }$ either by \eqref{sub_rtc} or by \eqref{sub_ode}. Let $m_c(t) =\E (S_c(t))$ and let the steady-state limit $m_{ \textnormal{eq} }(c)$ be given by \eqref{steadystatemean}. Then we have the following:
\begin{itemize}
\item[(A)] The value $m_{ \textnormal{eq} }(c)$ is well-defined (i.e. the limit in \eqref{steadystatemean} exists) and it is given by \eqref{mceqform1}.
\item[(B)] If $\tau_c$ is the random variable defined by \eqref{defn:tau} then \eqref{eqnetauc} holds.
\item[(C)] The limits \eqref{mdeterministic} and \eqref{mstatic} are satisfied as $c \to \infty$ and $c \to 0$ respectively.
\end{itemize}
\end{proposition}

\subsection{Expressions for $m_{ \textnormal{eq} }(c)$: The finite CTMC case} \label{sec:mainanalysisctmc}

In this section we specialize expression \eqref{mceqform1} to the case where $( \gamma(t) )_{t \geq 0}$ is a stationary CTMC with a finite state-space $\Gamma = \{ \gamma_1,\dots,\gamma_n \}$ as described in Section \ref{sec:model}. Define the CTMC $ ( \gamma_c(t) )_{t \geq 0}$ by \eqref{defngammac} and recall that its $n  \times n$ transition-rate matrix is given by $Q_c = c Q$. Let $D$ be the $n \times n$ diagonal matrix
\begin{linenomath*}
\begin{align}
\label{diagmatrixD}
D = \textnormal{Diag}(\gamma_1,\dots,\gamma_n)
\end{align}
\end{linenomath*}
whose entries are the degradation rates at different enzymatic states or abundance levels. One of the main results in our paper is to show that $m_{ \textnormal{eq} }(c) $ can be expressed
as
\begin{linenomath*}
\begin{align}
\label{mceqform2}
m_{ \textnormal{eq} }(c) = k_{ \textnormal{in} } \left[ \pi^T (D -c Q)^{-1} {\bf 1} \right].
\end{align}
\end{linenomath*}
Two alternative proofs of this result are given in the Supplementary Material. The first proof exploits some ideas from the theory of occupation measures for Markov chains \cite{GuptaEJP} while the second proof is based on the \emph{Methods of Conditional Moments} (MCM) approach recently developed by Hasenauer et al. \cite{Hasenauer13}. Note that this formula assumes matrix $(D -c Q)$ is invertible for any $c \geq 0$ but this can be easily verified from the properties of matrix $Q$. Using formula \eqref{mceqform2} we can prove that for any $c \geq 0$
\begin{linenomath*}
\begin{align}
\label{sandwichprop}
 m^{ \textnormal{(det)} }_{ \textnormal{eq} }  \leq m_{ \textnormal{eq} }(c)  \leq  m^{ \textnormal{(static)} }_{ \textnormal{eq} }.
\end{align}
\end{linenomath*}
Therefore for any relative speed $c$ of enzymatic fluctuations, the steady-state mean of the substrate is always sandwiched between the values obtained for the {\bf deterministic} and the {\bf static} cases. Moreover since $m_{ \textnormal{eq} }(c) $ depends continuously on $c$, the limits \eqref{mstatic} and \eqref{mdeterministic} imply that for any value $m^*$ in the open interval
$(  m^{ \textnormal{(det)} }_{ \textnormal{eq} } , m^{ \textnormal{(static)} }_{ \textnormal{eq} })$ there exists a relative speed value $c^* > 0$ such that
$m_{ \textnormal{eq} }(c^*) = m^*$. Hence the positive deviations caused by enzymatic fluctuations (in the mean substrate abundance) range from $0$ to \emph{exactly} $(m^{ \textnormal{(static)} }_{ \textnormal{eq} } - m^{ \textnormal{(det)} }_{ \textnormal{eq} } )$.

To detemine the map $c \mapsto m_{ \textnormal{eq} }(c)$ we need to evaluate $m_{ \textnormal{eq} }(c)$ at several values of $c$. This can be difficult with formula \eqref{mceqform2} because each evaluation requires inversion of a potentially large matrix. Fortunately we can resolve this issue using simple ideas from the \emph{theory of resolvents} for linear operators \cite{Kato}, as we now describe. Let $\tilde{Q} = D^{-1} Q$ be the transition rate matrix of another CTMC over state-space $\Gamma = \{ \gamma_1,\dots, \gamma_n \}$. The difference between this new CTMC and the original CTMC is that the rates of outflow from state $\gamma_i$ to each state $\gamma_j$ (for $j \neq i$) are divided by the state value $\gamma_i$. Let $\C$ denote the field of complex numbers. The resolvent for the Markov semigroup corresponding to this CTMC is the matrix-valued function over $\C$ defined by
\begin{linenomath*}
\begin{align}
\label{defn_resolv}
R(z) = (zI - \tilde{Q})^{-1},
\end{align}
\end{linenomath*}
where $I$ is the $n \times n$ Identity matrix. This function is well-defined for any $z$ which is not an eigenvalue of matrix $\tilde{Q}$. Let $\lambda_1,\dots,\lambda_n$ be the $n$ eigenvalues of matrix $\tilde{Q}$, repeated according to their algebraic multiplicity. Since $\tilde{Q}$ is the transition rate matrix of a CTMC, it has a simple\footnote{An eigenvalue is said to be simple if its algebraic multiplicity is $1$.} eigenvalue (say $\lambda_1$) equal to $0$, while its other eigenvalues have negative real parts. This implies that the resolvent function $R$ is well-defined on the positive real line $(0,\infty)$.

From now on we assume that matrix $\tilde{Q}$ is diagonalizable\footnote{A square matrix $M$ is diagonalizable if it can be written as $M = P \Lambda P^{-1}$ for some diagonal matrix $\Lambda$ and some invertible matrix $P$. The diagonal entries of $D$ are the eigenvalues of matrix $M$.} over the field $\C$ of complex numbers. This assumption is not very restrictive because almost every matrix is diagonalisable (see \cite{LinAl}) and so if $\tilde{Q}$ is not diagonalisable, we can perturb matrix $Q$ slightly to make $\tilde{Q}$ diagonalisable and not affect the enzyme dynamics significantly. The diagonalizability of $\tilde{Q}$ allows us to write matrix $\tilde{Q}$ as $\tilde{Q} = U \Lambda U^{-1}$,
where $\Lambda = \textnormal{Diag}( \lambda_1,\dots,\lambda_n)$ and $U$ is an invertible matrix whose columns contain the right eigenvectors for matrix $\tilde{Q}$ corresponding to the eigenvalues $\lambda_1,\dots,\lambda_n$. Similarly the rows of $U^{-1}$ contain the left eigenvectors for matrix $\tilde{Q}$ corresponding to the eigenvalues $\lambda_1,\dots,\lambda_n$. Let $u_i$ and $w_i$ be $n \times 1$ vectors denoting the $i$-th column and $i$-th row of matrices $U$ and $U^{-1}$ respectively. Therefore
\begin{linenomath*}
\begin{align}
\label{qtildeexpansion}
\tilde{Q} = \sum_{i=1}^n \lambda_i u_i w^T_i
\end{align}
\end{linenomath*}
and we can express the resolvent function $R$ (see Chapter 5 in \cite{Kato}) as
\begin{linenomath*}
\begin{align}
\label{defn_resolv2}
R(z) = \sum_{i=1}^n  \left( \frac{1}{z - \lambda_i } \right) u_i w^T_i
\end{align}
\end{linenomath*}
Let $ \langle \cdot, \cdot \rangle$ denote the standard inner product on $\R^n$. Note that formula \eqref{mceqform2} can be expressed as
\begin{linenomath*}
\begin{align*}
m_{ \textnormal{eq} }(c) = k_{ \textnormal{in} } \left[ \pi^T (I -c \tilde{Q})^{-1} D^{-1}{\bf 1} \right] = c^{-1} k_{ \textnormal{in} } \left[ \pi^T R(c^{-1}) D^{-1}{\bf 1} \right]
\end{align*}
\end{linenomath*}
for any $c >0$. Plugging $R(c^{-1})$ from \eqref{defn_resolv2} and defining
\begin{linenomath*}
\begin{align}
\label{defn_weights}
\alpha_i =  \langle \pi, u_i \rangle  \langle w_i,  D^{-1}{\bf 1}  \rangle = (\pi^T u_i) ( w^T_i  D^{-1}{\bf 1}  )  \qquad \textnormal{ for each }\qquad i=1,\dots,n
\end{align}
\end{linenomath*}
we obtain the following formula for $m_{ \textnormal{eq} }(c)$:
\begin{linenomath*}
\begin{align}
\label{mainmceqformula}
m_{ \textnormal{eq} }(c) = k_{ \textnormal{in} } \sum_{i=1}^n  \left( \frac{ \alpha_i }{1 - c \lambda_i } \right).
\end{align}
\end{linenomath*}
Observe that since $\alpha_i$-s and $\lambda_i$-s are independent of $c$, they only need to be computed once to construct this expression and then we can easily compute $m_{ \textnormal{eq} }(c)$ for several values of $c$ without the need of evaluating the matrix inverses in \eqref{mceqform2}. Moreover if $n$ is large, then using the values of $\alpha_i$ and $\lambda_i$ as a guide, one can derive suitable approximations of the formula \eqref{mainmceqformula} for $m_{ \textnormal{eq} }(c)$. We derive one such approximation in the next section and use it as a tool to further understand the phenomenon of stochastic amplification induced by dynamic disorder in enzymatic activity.

The results from this section are collected in our next theorem which is proved in the Supplementary Material.
\begin{theorem}
\label{theorem_general}
Suppose $ ( \gamma(t) )_{t \geq 0 }$ is a stationary CTMC with transition rate matrix $Q$, stationary distribution $\pi$ and state-space $\Gamma = \{ \gamma_1,\dots,\gamma_n \}$ (see Section \ref{sec:model}). Let $ ( \gamma_c(t) )_{t \geq 0}$ be the speed $c$ version of this process given by \eqref{defngammac} and define the substrate dynamics $( S_c(t) )_{t \geq 0 }$ either by \eqref{sub_rtc} or by \eqref{sub_ode}. Let the steady-state substrate mean $m_{ \textnormal{eq} }(c)$ be given by \eqref{steadystatemean} and the diagonal matrix $D$ be defined by \eqref{diagmatrixD}. Then we have the following:
\begin{itemize}
\item[(A)] The matrix $(D -c Q)$ is invertible and $m_{ \textnormal{eq} }(c)$ can be expressed as \eqref{mceqform2}.
\item[(B)] Suppose the matrix $\tilde{Q} = D^{-1} Q$ is diagonalizable and let $\lambda_1,\dots,\lambda_n$ be its eigenvalues. For each $i=1,\dots,n$ define $\alpha_i$ by \eqref{defn_weights}. Then $m_{ \textnormal{eq} }(c)$ can be expressed as \eqref{mainmceqformula}.
\item[(C)] The relation \eqref{sandwichprop} is satisfied for any $c \geq 0$.
\end{itemize}
\end{theorem}

\subsection{Approximate formula for $m_{ \textnormal{eq} }(c)$} \label{sec:approxformula}

The goal of this section is to derive an approximate formula for $m_{ \textnormal{eq} }(c)$ using \eqref{mainmceqformula} and then use it to obtain some interesting insights. Recall from the previous section that $\lambda_1,\dots,\lambda_n$ are the eigenvalues of matrix $\tilde{Q}$. Among these $\lambda_1 = 0$ while the eigenvalues $\lambda_2,\dots,\lambda_n$ have negative real parts. Define a positive constant $\epsilon_{ \textnormal{max} }$ by
\begin{linenomath*}
\begin{align*}
\epsilon_{ \textnormal{max} } = - \max\{ \textnormal{Re}( \lambda_i ) : i=2,\dots,n  \},
\end{align*}
\end{linenomath*}
where $\textnormal{Re}(z)$ denotes the real part of a complex number $z$. Setting $\lambda_1 = 0$ in \eqref{mainmceqformula} we obtain
\begin{linenomath*}
\begin{align}
\label{resolventformula}
m_{ \textnormal{eq} }(c) = k_{ \textnormal{in} } \left[  \alpha_1 +   \sum_{i=2}^n  \left( \frac{ \alpha_i }{1 - c \lambda_i } \right)   \right].
\end{align}
\end{linenomath*}
This formula is valid for any $c$ in the interval $(  -\epsilon_{ \textnormal{max} }  ,\infty)$ and its form shows that the function $m_{ \textnormal{eq} }(c)$ is \emph{real-analytic}\footnote{A function is called real analytic at a point if it is infinitely differentiable at that point and it agrees with its Taylor series expansion around that point.} at $c=0$. Therefore all the information about function $m_{ \textnormal{eq} }(c)$ is contained in the value of this function and its derivatives at $c=0$.

Using limits \eqref{mdeterministic} and \eqref{mstatic} we can conclude that
\begin{linenomath*}
\begin{align}
\label{alpharelations}
\alpha_1=  \frac{ m^{ \textnormal{(det)} }_{ \textnormal{eq} } }{ k_{ \textnormal{in} } }   \qquad \textnormal{and} \qquad \sum_{i=2}^n  \alpha_i = \left( \frac{m^{ \textnormal{(static)} }_{ \textnormal{eq} } - m^{ \textnormal{(det)} }_{ \textnormal{eq} } }{k_{ \textnormal{in} }}\right).
\end{align}
\end{linenomath*}
%
Let $\theta$ denote the following weighted combination of eigenvalues $\lambda_2,\dots,\lambda_n$
\begin{linenomath*}
\begin{align}
\label{defn_hatlambda}
\theta = -  \frac{  \sum_{i=2}^n \lambda_i \alpha_i }{\sum_{i=2}^n \alpha_i}.
\end{align}
\end{linenomath*}
We now propose an approximate formula for $m_{ \textnormal{eq} }(c)$
\begin{linenomath*}
\begin{align}
\label{mceqapprox}
\hat{m}_{ \textnormal{eq} }(c) = m^{ \textnormal{(det)} }_{ \textnormal{eq} } +   \left(  \frac{m^{ \textnormal{(static)} }_{ \textnormal{eq} } - m^{ \textnormal{(det)} }_{ \textnormal{eq} }}{1 + c \theta}  \right) .
\end{align}
\end{linenomath*}
Note this formula is much easier to use than \eqref{resolventformula} because it contains only one rational term. From \eqref{alpharelations} it is immediate that $\hat{m}_{ \textnormal{eq} }(c)$ also obeys the limits \eqref{mdeterministic} and \eqref{mstatic}. Moreover it is straightforward to check that the first derivatives of $\hat{m}_{ \textnormal{eq} }(c) $ and $m_{ \textnormal{eq} }(c)$ match at $c = 0$. Hence the approximation error is given by a difference of second-order derivatives and we explain in Supplementary Material why this error is likely to be small. We can also view the approximation $\hat{m}_{ \textnormal{eq} }(c) $ of $m_{ \textnormal{eq} }(c)$ as replacing a weighted \emph{arithmetic mean} of several quantities with the corresponding \emph{harmonic mean}. To see this note that from \eqref{resolventformula} and \eqref{alpharelations} we can express $m_{ \textnormal{eq} }(c) $ as
\begin{linenomath*}
\begin{align}
\label{newresolvent}
m_{ \textnormal{eq} }(c) = m^{ \textnormal{(det)} }_{ \textnormal{eq} } +  \left(  m^{ \textnormal{(static)} }_{ \textnormal{eq} } - m^{ \textnormal{(det)} }_{ \textnormal{eq} } \right) \bar{x},
\end{align}
\end{linenomath*}
where
\begin{linenomath*}
\begin{align*}
\bar{x} =  \frac{ \sum_{i=2}^n  \alpha_i  (1 - c \lambda_i)^{-1}  }{ \sum_{i=2}^n  \alpha_i  }
\end{align*}
\end{linenomath*}
is the weighted arithmetic mean of quantities $(1 - c \lambda_2)^{-1},\dots,(1 - c \lambda_n)^{-1}$ with weights $\alpha_2,\dots,\alpha_n$\footnote{These weights may not be positive real numbers, as is customary in the definition of arithmetic means. However in our examples we generally find that the most \emph{significant} weights indeed have a positive real part and a negligible imaginary part.}. The corresponding weighted harmonic mean of these quantities is given by
\begin{linenomath*}
\begin{align*}
\hat{x} = \frac{ \sum_{i=2}^n  \alpha_i  }{  \sum_{i=2}^n \alpha_i (1 - c \lambda_i) } = \frac{1}{1 + c \theta }
\end{align*}
\end{linenomath*}
and observe that $\hat{m}_{ \textnormal{eq} }(c)$ can be expressed as the r.h.s. of \eqref{newresolvent} with arithmetic mean $\bar{x}$ replaced by the harmonic mean $\hat{x}$.

 We now illustrate the accuracy of this approximation using a couple of randomly generated $n \times n$, transition rate matrices $Q$ with $n=5$ and $n=10$ respectively. In both cases we choose the input rate to be $k_{ \textnormal{in} }  =1$ and the enzymatic state-values to be $\gamma_i = i$ for $1,2,\dots,n$. The exact function $\hat{m}_{ \textnormal{eq} }(c) $ along with its approximation $\hat{m}_{ \textnormal{eq} }(c)$ are plotted in Figure \ref{fig:approximation}. The accuracy of this approximation can be easily seen. Notice that the exact function is \emph{slightly above} its approximation. Assuming the significant weights ($\alpha_i$-s) are positive  reals, this can be explained by the fact that arithmetic mean is always higher than the corresponding harmonic mean.

From \eqref{mceqapprox} it is immediate that the shape of the function $\hat{m}_{ \textnormal{eq} }(c) $ depends crucially on the parameter $\theta $ computed according to \eqref{defn_hatlambda}. We now examine $\theta $ more closely and see how it is connected to an existing notion from the theory of Markov processes. Let us denote the numerator of \eqref{defn_hatlambda} by
\begin{linenomath*}
\begin{align}
\label{defn_theta}
\Theta = - \sum_{i=2}^n \lambda_i \alpha_i.
\end{align}
\end{linenomath*}
Since $\alpha_i$-s are given by \eqref{defn_weights}, using \eqref{qtildeexpansion}, $\tilde{Q} = D^{-1}Q$ and $\lambda_1 = 0$ we can express $\Theta$ as
\begin{linenomath*}
\begin{align}
\label{connect_dirform1}
\Theta = -\sum_{i=1}^n \lambda_i \alpha_i  = -\pi^T  \left( \sum_{i=1}^n \lambda_i u_i w^T_i \right)  D^{-1} {\bf 1} = -\pi^T \tilde{Q}  D^{-1} {\bf 1} = -\pi^T D^{-1}  Q  D^{-1} {\bf 1}.
\end{align}
\end{linenomath*}
This relation shows that $\Theta$ (and hence $\theta $) is always real-valued even though some $\lambda_i$-s or $\alpha_i$-s may have imaginary parts. Moreover to compute $\Theta$ we do not need to compute the eigenvalues $\lambda_1,\dots,\lambda_n$ of a potentially large matrix $\tilde{Q}$. Instead we only need to evaluate the expression $\pi^T D^{-1}  Q  D^{-1} {\bf 1}$ which is computationally much easier. Interestingly the definition of $\Theta$ coincides with the well-known notion of \emph{Dirichlet forms}, that is extensively used in the study of mixing properties of Markov processes \cite{Saloff-Coste,PeresLevin}. We now discuss this connection in more detail.

Consider the CTMC $( \gamma(t) )_{t \geq 0}$ with state-space $\Gamma = \{ \gamma_1,\dots, \gamma_n\}$ and transition rate matrix $Q = [q_{ij}]$. The generator $\mathbb{Q}$\footnote{The generator of a Markov process is an operator which specifies the rate of change of the distribution of the process. For more details see Chapter 4 in \cite{EK}.} of this CTMC maps any real-valued function $f$ on $\Gamma$ to another such real-valued function $\mathbb{Q}f$ given by
\begin{linenomath*}
\begin{align*}
\mathbb{Q} f( \gamma_i ) = \sum_{j  \neq i} q_{ij} ( f( \gamma_j ) - f( \gamma_i ) ).
\end{align*}
\end{linenomath*}
Define a function $f : \Gamma \to (0,\infty)$ by $f( \gamma ) = 1/\gamma$. Then one can see that $\Theta$ \eqref{defn_theta} can be expressed as
\begin{linenomath*}
\begin{align}
\label{thetaisdif}
\Theta = - \E_{\pi} ( f(\gamma)\mathbb{Q} f( \gamma)  ).
\end{align}
\end{linenomath*}
In other words, if a $\Gamma$-valued random variable $\gamma$ has distribution $\pi$, then $\Theta$ is the expectation of the random variable $( - f(\gamma)\mathbb{Q} f( \gamma) )$. Relation \eqref{thetaisdif} shows that $\Theta$ is a Dirichlet form associated with the Markovian semigroup generated by $\mathbb{Q}$ (see \cite{Saloff-Coste}). An important consequence of this connection is that $\Theta$ is always positive (see Lemma 2.1.2 in \cite{Saloff-Coste}) irrespective of the entries of the rate matrix $Q$ or the state values $\gamma_1,\dots,\gamma_n$. The positivity of $\Theta$ implies that $\theta $ is also positive and hence the mapping $c \mapsto \hat{m}_{ \textnormal{eq} }(c) $ is convex and monotonically decreasing from $m^{ \textnormal{(static)} }_{ \textnormal{eq} }$ at $c =0$ to $m^{ \textnormal{(det)} }_{ \textnormal{eq} }$ as $c \to \infty$. Intuitively the magnitude of Dirichlet form $\Theta$ corresponds to the \emph{mixing strength} of the underlying Markov process. Therefore as $\Theta$ increases, $\theta $ also increases and the mapping $c \mapsto \hat{m}_{ \textnormal{eq} }(c) $ has a sharper ``drop" to the deterministic value $m^{ \textnormal{(det)} }_{ \textnormal{eq} }$. Our next goal is to make this mathematically precise and quantitatively estimate the relative speed-values $c$ beyond which the deterministic assumption is acceptable.

In the rest of this section, our object of interest will be the \emph{relative stochastic amplification factor} defined by
\begin{linenomath*}
\begin{align}
\label{defn_stochampl}
\rho(c) = \frac{m_{ \textnormal{eq} }(c) - m^{ \textnormal{(det)} }_{ \textnormal{eq} }  }{ m^{ \textnormal{(det)} }_{ \textnormal{eq} } } ,
\end{align}
\end{linenomath*}
which measures the difference of steady-state substrate means in the presence and absence of enzymatic fluctuations, \emph{normalized} by the the steady-state substrate mean in the deterministic case. Note that $\rho(c)$ does not depend on the input rate $k_{ \textnormal{in} }$ and using \eqref{sandwichprop} we can see that $\rho(c)$ satisfies
\begin{linenomath*}
\begin{align}
\label{defn_rho_max}
0 \leq \rho(c) \leq \rho_{ \textnormal{max} } :=  \frac{m^{ \textnormal{(static)} }_{ \textnormal{eq} }  }{m^{ \textnormal{(det)} }_{ \textnormal{eq} }} - 1 = \E_\pi\left( \frac{1}{\gamma} \right) \E_\pi( \gamma ) - 1 \quad \textnormal{for any} \quad c \geq 0.
\end{align}
\end{linenomath*}
In order to study the dependence of $\rho(c)$ on $c$, we now look at its approximation $\hat{\rho}(c)$ which is defined analogously to \eqref{defn_stochampl}, with $m_{ \textnormal{eq} }(c)$ replaced by $\hat{m}_{ \textnormal{eq} }(c)$. Using \eqref{defn_hatlambda}, \eqref{defn_theta} and \eqref{alpharelations} we see that $\theta$ is the same as the \emph{normalized} Dirichlet form defined by
\begin{linenomath*}
\begin{align}
\label{thetaisdif2}
\theta = \frac{ \Theta  \E_\pi( \gamma ) }{ \rho_{ \textnormal{max} }   }.
\end{align}
\end{linenomath*}
Substituting $\hat{ \lambda }$ by $\theta$ in \eqref{mceqapprox} and dividing by $m^{ \textnormal{(det)} }_{ \textnormal{eq} }$, we obtain the following formula after some simple algebraic manipulations
\begin{linenomath*}
\begin{align}
\label{main_ampl_formula}
\hat{\rho}(c) = \frac{  \rho_{ \textnormal{max} }  }{ 1 + c\theta  } .
\end{align}
\end{linenomath*}
This formula clearly indicates that as $\theta$ gets larger, the amplification factor decreases more sharply to $1$ with the relative speed parameter $c$. One can regard $ \rho(c) \approx \hat{\rho}(c)$ as the ``relative error" between the actual substrate mean and the mean computed with deterministic assumption on the enzymatic kinetics. From relation \eqref{main_ampl_formula} it is immediate that in order to \emph{test} if this error will exceed some tolerance level $\epsilon >0$ we just need to check if the relative enzyme speed $c$ is smaller than the threshold $c_{\epsilon}$ defined by
\begin{linenomath*}
\begin{align}
\label{cthreshold}
c_\epsilon:=  \left( \frac{  \rho_{ \textnormal{max} }  -\epsilon  }{ \theta \epsilon }  \right).
\end{align}
\end{linenomath*}
We can expect this test to be rather conservative because as we have argued before, the exact values $m_{ \textnormal{eq} }(c)$ will usually lie above their approximation $\hat{m}_{ \textnormal{eq} }(c)$.

Note that $c_\epsilon$ is inversely proportional to $\theta$ but directly proportional to $\rho_{ \textnormal{max} }$. The first parameter $\theta$ is the normalized Dirichlet form and it captures the ``mixing strength" of the underlying enzymatic dynamics (see Example \ref{example:twostate}), while the second parameter $\rho_{ \textnormal{max} }$ can be viewed as a proxy for the variance of the stationary-distribution $\pi$\footnote{To see this note that $\rho_{ \textnormal{max} }$ defined by \eqref{defn_rho_max} represents the ``error" in Jensen's inequality for the convex map $x \mapsto 1/x$. It can be easily shown that this error is proportional to the variance of the distribution $\pi$ (see \cite{Costarelli2015} for instance).}(see Example \ref{example:twostate}). Generally both these parameters will increase with higher levels \emph{noise} in the enzymatic dynamics. However since they affect $c_\epsilon$ in opposing ways, it is difficult to ascertain the overall effect of dynamical noise in setting the threshold value $c_\epsilon$. We explore this issue in greater detail in Section \ref{example:sf} and numerically show that increasing levels of dynamical noise in the enzymatic kinetics of that reaction network gives rise to decreasing values of $c_\epsilon$. This is surprising and counterintuitive because it suggests that this dynamical noise is actually \emph{beneficial} in improving the accuracy of the deterministic assumption for the enzyme activity.

Finally we remark that even though most of the analysis in this paper assumes that enzymatic kinetics is described by a finite Markov chain, the formulas we derive can provide insights for a more general class of stationary stochastic processes. This is because finite Markov chains can serve as good approximations of such processes \cite{Pincus}. Moreover if the process is Markov, even with an arbitrary state-space, we can compute expression \eqref{main_ampl_formula} for $\hat{\rho}(c)$ by sampling its stationary distribution and using this sample to estimate $\rho_{ \textnormal{max} }$ and the normalized Dirichlet form $\theta$. We illustrate this for the example network in Section \ref{example:sf} where the enzyme dynamics follows a Markov process over a countable state-space.

\begin{figure}[!h]
\begin{center}
    \includegraphics[width=1\textwidth]{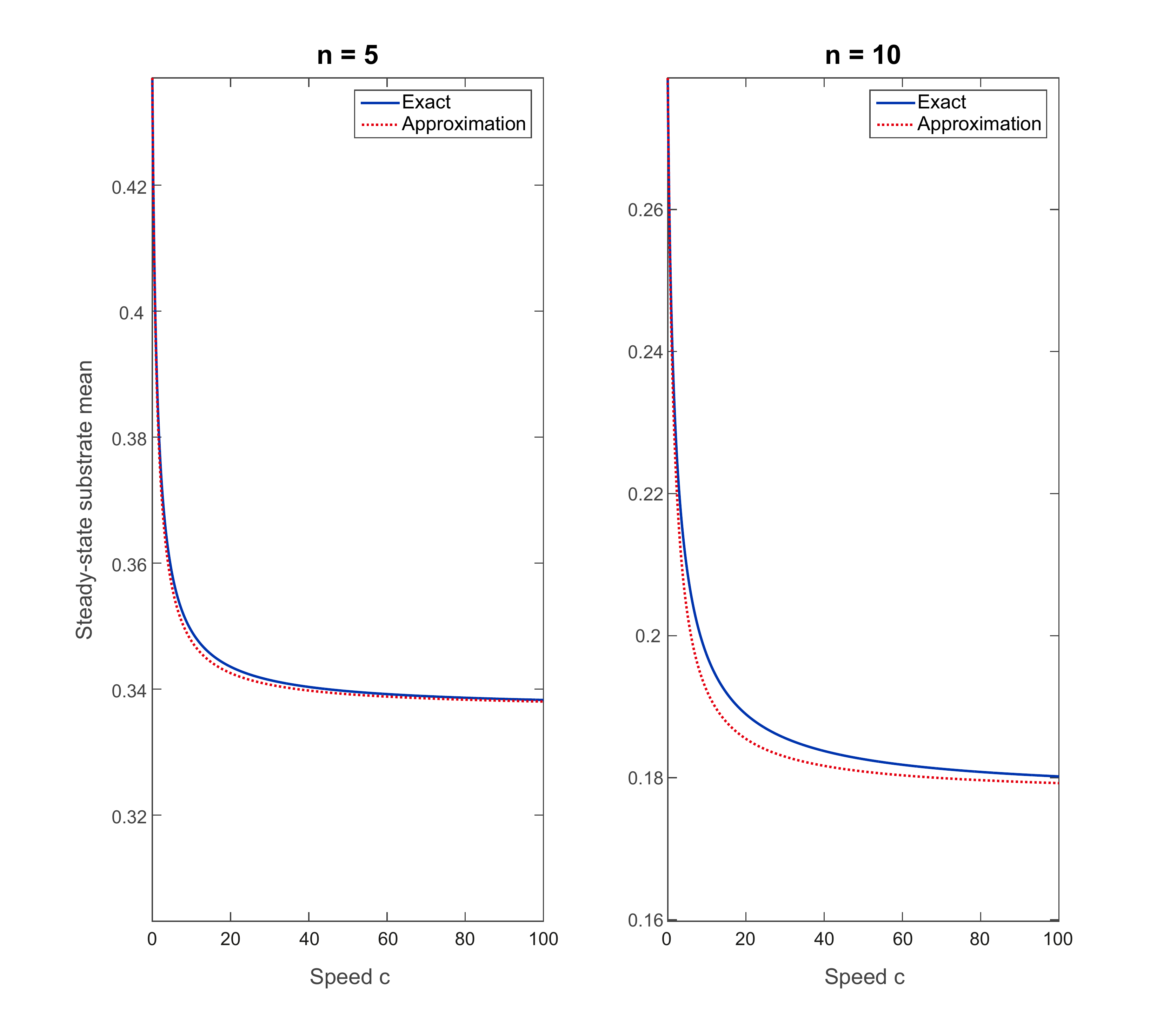}
\end{center}
    \caption{Comparison of the exact steady-state mean substrate value $m_{ \textnormal{eq} }(c)$ with its approximation $\hat{m}_{ \textnormal{eq} }(c) $}
    \label{fig:approximation}
\end{figure}

From \eqref{mceqapprox} it is immediate that the shape of the function $\hat{m}_{ \textnormal{eq} }(c) $ depends crucially on the parameter $\theta $ computed according to \eqref{defn_hatlambda}. We now examine $\theta $ more closely and see how it is connected to an existing notion from the theory of Markov processes. Let us denote the numerator of \eqref{defn_hatlambda} by
\begin{linenomath*}
\begin{align}
\label{defn_theta}
\Theta = - \sum_{i=2}^n \lambda_i \alpha_i.
\end{align}
\end{linenomath*}
Since $\alpha_i$-s are given by \eqref{defn_weights}, using \eqref{qtildeexpansion}, $\tilde{Q} = D^{-1}Q$ and $\lambda_1 = 0$ we can express $\Theta$ as
\begin{linenomath*}
\begin{align}
\label{connect_dirform1}
\Theta = -\sum_{i=1}^n \lambda_i \alpha_i  = -\pi^T  \left( \sum_{i=1}^n \lambda_i u_i w^T_i \right)  D^{-1} {\bf 1} = -\pi^T \tilde{Q}  D^{-1} {\bf 1} = -\pi^T D^{-1}  Q  D^{-1} {\bf 1}.
\end{align}
\end{linenomath*}
This relation shows that $\Theta$ (and hence $\theta $) is always real-valued even though some $\lambda_i$-s or $\alpha_i$-s may have imaginary parts. Moreover to compute $\Theta$ we do not need to compute the eigenvalues $\lambda_1,\dots,\lambda_n$ of a potentially large matrix $\tilde{Q}$. Instead we only need to evaluate the expression $\pi^T D^{-1}  Q  D^{-1} {\bf 1}$ which is computationally much easier. Interestingly the definition of $\Theta$ coincides with the well-known notion of \emph{Dirichlet forms}, that is extensively used in the study of mixing properties of Markov processes \cite{Saloff-Coste,PeresLevin}. We now discuss this connection in more detail.

Consider the CTMC $( \gamma(t) )_{t \geq 0}$ with state-space $\Gamma = \{ \gamma_1,\dots, \gamma_n\}$ and transition rate matrix $Q = [q_{ij}]$. The generator $\mathbb{Q}$\footnote{The generator of a Markov process is an operator which specifies the rate of change of the distribution of the process. For more details see Chapter 4 in \cite{EK}.} of this CTMC maps any real-valued function $f$ on $\Gamma$ to another such real-valued function $\mathbb{Q}f$ given by
\begin{linenomath*}
\begin{align*}
\mathbb{Q} f( \gamma_i ) = \sum_{j  \neq i} q_{ij} ( f( \gamma_j ) - f( \gamma_i ) ).
\end{align*}
\end{linenomath*}
Define a function $f : \Gamma \to (0,\infty)$ by $f( \gamma ) = 1/\gamma$. Then one can see that $\Theta$ \eqref{defn_theta} can be expressed as
\begin{linenomath*}
\begin{align}
\label{thetaisdif}
\Theta = - \E_{\pi} ( f(\gamma)\mathbb{Q} f( \gamma)  ).
\end{align}
\end{linenomath*}
In other words, if a $\Gamma$-valued random variable $\gamma$ has distribution $\pi$, then $\Theta$ is the expectation of the random variable $( - f(\gamma)\mathbb{Q} f( \gamma) )$. Relation \eqref{thetaisdif} shows that $\Theta$ is a Dirichlet form associated with the Markovian semigroup generated by $\mathbb{Q}$ (see \cite{Saloff-Coste}). An important consequence of this connection is that $\Theta$ is always positive (see Lemma 2.1.2 in \cite{Saloff-Coste}) irrespective of the entries of the rate matrix $Q$ or the state values $\gamma_1,\dots,\gamma_n$. The positivity of $\Theta$ implies that $\theta $ is also positive and hence the mapping $c \mapsto \hat{m}_{ \textnormal{eq} }(c) $ is convex and monotonically decreasing from $m^{ \textnormal{(static)} }_{ \textnormal{eq} }$ at $c =0$ to $m^{ \textnormal{(det)} }_{ \textnormal{eq} }$ as $c \to \infty$. Intuitively the magnitude of Dirichlet form $\Theta$ corresponds to the \emph{mixing strength} of the underlying Markov process. Therefore as $\Theta$ increases, $\theta $ also increases and the mapping $c \mapsto \hat{m}_{ \textnormal{eq} }(c) $ has a sharper ``drop" to the deterministic value $m^{ \textnormal{(det)} }_{ \textnormal{eq} }$. Our next goal is to make this mathematically precise and quantitatively estimate the relative speed-values $c$ beyond which the deterministic assumption is acceptable.

In the rest of this section, our object of interest will be the \emph{relative stochastic amplification factor} defined by
\begin{linenomath*}
\begin{align}
\label{defn_stochampl}
\rho(c) = \frac{m_{ \textnormal{eq} }(c) - m^{ \textnormal{(det)} }_{ \textnormal{eq} }  }{ m^{ \textnormal{(det)} }_{ \textnormal{eq} } } ,
\end{align}
\end{linenomath*}
which measures the difference of steady-state substrate means in the presence and absence of enzymatic fluctuations, \emph{normalized} by the the steady-state substrate mean in the deterministic case. Note that $\rho(c)$ does not depend on the input rate $k_{ \textnormal{in} }$ and using \eqref{sandwichprop} we can see that $\rho(c)$ satisfies
\begin{linenomath*}
\begin{align}
\label{defn_rho_max}
0 \leq \rho(c) \leq \rho_{ \textnormal{max} } :=  \frac{m^{ \textnormal{(static)} }_{ \textnormal{eq} }  }{m^{ \textnormal{(det)} }_{ \textnormal{eq} }} - 1 = \E_\pi\left( \frac{1}{\gamma} \right) \E_\pi( \gamma ) - 1 \quad \textnormal{for any} \quad c \geq 0.
\end{align}
\end{linenomath*}
In order to study the dependence of $\rho(c)$ on $c$, we now look at its approximation $\hat{\rho}(c)$ which is defined analogously to \eqref{defn_stochampl}, with $m_{ \textnormal{eq} }(c)$ replaced by $\hat{m}_{ \textnormal{eq} }(c)$. Using \eqref{defn_hatlambda}, \eqref{defn_theta} and \eqref{alpharelations} we see that $\theta$ is the same as the \emph{normalized} Dirichlet form defined by
\begin{linenomath*}
\begin{align}
\label{thetaisdif2}
\theta = \frac{ \Theta  \E_\pi( \gamma ) }{ \rho_{ \textnormal{max} }   }.
\end{align}
\end{linenomath*}
Dividing \eqref{mceqapprox} by $m^{ \textnormal{(det)} }_{ \textnormal{eq} }$, we obtain the following formula after some simple algebraic manipulations
\begin{linenomath*}
\begin{align}
\label{main_ampl_formula}
\hat{\rho}(c) = \frac{  \rho_{ \textnormal{max} }  }{ 1 + c\theta  } .
\end{align}
\end{linenomath*}
This formula clearly indicates that as $\theta$ gets larger, the amplification factor decreases more sharply to $1$ with the relative speed parameter $c$. One can regard $ \rho(c) \approx \hat{\rho}(c)$ as the ``relative error" between the actual substrate mean and the mean computed with deterministic assumption on the enzymatic kinetics. From relation \eqref{main_ampl_formula} it is immediate that in order to \emph{test} if this error will exceed some tolerance level $\epsilon >0$ we just need to check if the relative enzyme speed $c$ is smaller than the threshold $c_{\epsilon}$ defined by
\begin{linenomath*}
\begin{align}
\label{cthreshold}
c_\epsilon:=  \left( \frac{  \rho_{ \textnormal{max} }  -\epsilon  }{ \theta \epsilon }  \right).
\end{align}
\end{linenomath*}
We can expect this test to be rather conservative because as we have argued before, the exact values $m_{ \textnormal{eq} }(c)$ will usually lie above their approximation $\hat{m}_{ \textnormal{eq} }(c)$.

Note that $c_\epsilon$ is inversely proportional to $\theta$ but directly proportional to $\rho_{ \textnormal{max} }$. The first parameter $\theta$ is the normalized Dirichlet form and it captures the ``mixing strength" of the underlying enzymatic dynamics (see Example \ref{example:twostate}), while the second parameter $\rho_{ \textnormal{max} }$ can be viewed as a proxy for the variance of the stationary-distribution $\pi$\footnote{To see this note that $\rho_{ \textnormal{max} }$ defined by \eqref{defn_rho_max} represents the ``error" in Jensen's inequality for the convex map $x \mapsto 1/x$. It can be easily shown that this error is proportional to the variance of the distribution $\pi$ (see \cite{Costarelli2015} for instance).}(see Example \ref{example:twostate}). Generally both these parameters will increase with higher levels \emph{noise} in the enzymatic dynamics. However since they affect $c_\epsilon$ in opposing ways, it is difficult to ascertain the overall effect of dynamical noise in setting the threshold value $c_\epsilon$. We explore this issue in greater detail in Section \ref{example:sf} and numerically show that increasing levels of dynamical noise in the enzymatic kinetics of that reaction network gives rise to decreasing values of $c_\epsilon$. This is surprising and counterintuitive because it suggests that this dynamical noise is actually \emph{beneficial} in improving the accuracy of the deterministic assumption for the enzyme activity.

Finally we remark that even though most of the analysis in this paper assumes that enzymatic kinetics is described by a finite Markov chain, the formulas we derive can provide insights for a more general class of stationary stochastic processes. This is because finite Markov chains can serve as good approximations of such processes \cite{Pincus}. Moreover if the process is Markov, even with an arbitrary state-space, we can compute expression \eqref{main_ampl_formula} for $\hat{\rho}(c)$ by sampling its stationary distribution and using this sample to estimate $\rho_{ \textnormal{max} }$ and the normalized Dirichlet form $\theta$. We illustrate this for the example network in Section \ref{example:sf} where the enzyme dynamics follows a Markov process over a countable state-space.

\section{Examples}   \label{sec:numexampl}

In this section we present a couple of examples to illustrate our results. Our first example of a two-state switching enzyme is such that all calculations can be easily done analytically allowing us to clearly understand the stochastic amplification effect. We also see how enzymes can utilize their fluctuations to serve as high-gain amplifiers. Our second example is the reaction network of Paulsson et al. \cite{Paulsson00} which displays stochastic focusing. We apply our results to this network and demonstrate that in some cases dynamical fluctuations can actually be beneficial in reducing the unwanted stochastic amplification effects.

\subsection{Stochastic amplification induced by a two-state switching enzyme} \label{example:twostate}

Consider a simple instance of the system in Figure \ref{fig:model}, in which a single enzyme molecule is present and can fluctuate between two states of activity: with enzyme ${\bf E}$ in the low-activity (``0'') state, the degradation rate of substrate ${\bf S}$ is assumed to be $\gamma_0$, while it is equal to $\gamma_1$ when ${\bf E}$ is highly active (``1''). The 0-to-1 and 1-to-0 rates are given by $k_{ \textnormal{on} }$ and $k_{ \textnormal{off} }$ respectively. A schematic representation of this model is presented in Figure \ref{fig:twostatemodel}.
\begin{figure}[h]
\begin{center}
    \includegraphics[width=0.6\textwidth]{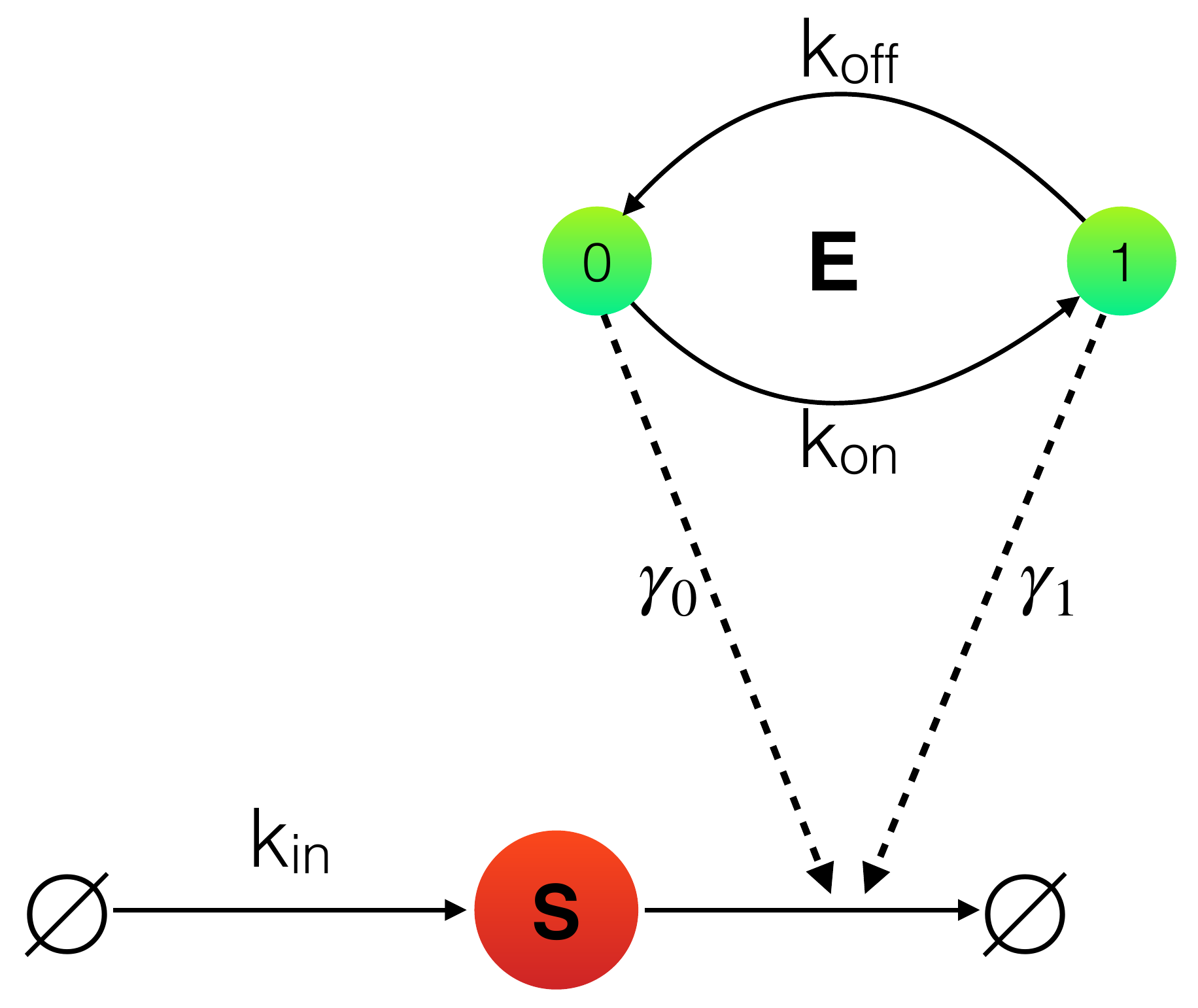}
\end{center}
    \caption{Schematic of the two-state switching enzyme model}
    \label{fig:twostatemodel}
\end{figure}

The time-varying degradation rate $( \gamma(t) )_{t \geq 0}$ induced by this fluctuating enzyme ${\bf E}$  is a CTMC with state-space $\Gamma = \{ \gamma_0, \gamma_1 \}$ and transition-rate matrix
\begin{linenomath*}
\begin{align*}
Q = \left[
\begin{array}{cc}
-k_{ \textnormal{on} }  & k_{ \textnormal{on} }\\
k_{ \textnormal{off} } & -k_{ \textnormal{off} }
\end{array} \right].
\end{align*}
\end{linenomath*}
One can check that the unique stationary distribution $\pi = (\pi_0, \pi_1)$ for this CTMC is simply given by
\begin{linenomath*}
\begin{align}
\label{defn_stdis_twostateenyme}
\pi_0 = \frac{k_{ \textnormal{off} } }{k_{ \textnormal{on} }+k_{ \textnormal{off} }} \qquad \textnormal{and} \qquad  \pi_1 = \frac{k_{ \textnormal{on} } }{k_{ \textnormal{on} }+k_{ \textnormal{off} }}.
\end{align}
\end{linenomath*}
For each $i=0,1$, we can regard $\pi_i$ as the steady-state probability of the enzyme being in state $i$. Due to the Ergodic Theorem (see Theorem 10.6 in \cite{Kal}) we can also view $\pi_i$ as the proportion of time that the enzyme spends in state $i$ in the long-run. Let $\gamma$ be a $\Gamma$-valued random variable with probability distribution $\pi$. Then its mean and variance can be computed as
\begin{linenomath*}
\begin{align}
\label{twostateenzyme:mean_and_var}
\E_\pi( \gamma ) = \left( \frac{k_{ \textnormal{off} } \gamma_0 +  k_{ \textnormal{on} }  \gamma_1  }{k_{ \textnormal{on} }+k_{ \textnormal{off} }} \right) \qquad \textnormal{and} \qquad \textnormal{Var}_\pi( \gamma ) =  \frac{  k_{ \textnormal{on} }  k_{ \textnormal{off} }   }{ (k_{ \textnormal{on} }+k_{ \textnormal{off} } )^2  }  ( \gamma_0  - \gamma_1 )^2.
\end{align}
\end{linenomath*}

Suppose that the speed of enzymatic kinetics relative to the substrate is $c$ and so the degradation rate is given by the process $( \gamma_c(t) )_{t  \geq 0 }$ defined by \eqref{defngammac}. Let $m_{ \textnormal{eq} }(c)$ \eqref{steadystatemean} be the steady-state substrate mean in this case. Using part (A) of Theorem \ref{theorem_general} we obtain
\begin{linenomath*}
\begin{align*}
m_{ \textnormal{eq} }(c) =  k_{ \textnormal{in} }
[
\begin{array}{cc}
\pi_0 & \pi_1
\end{array}
]
\left[
\begin{array}{cc}
\gamma_0 + c k_{ \textnormal{on} } & -  c k_{ \textnormal{on} }\\
 - c k_{ \textnormal{off} } &  \gamma_1 + c k_{ \textnormal{off} }
\end{array}
\right]^{-1}
\left[
\begin{array}{c}
1 \\
1
\end{array}
\right].
\end{align*}
\end{linenomath*}
This formula involves the inverse of a $2 \times 2$ matrix, which can be easily computed explicitly. Substituting this inverse along with the expressions for $\pi_0$ and $\pi_1$ (see \eqref{defn_stdis_twostateenyme}) we get
\begin{linenomath*}
\begin{align}
\label{main_formula_two_enz}
m_{ \textnormal{eq} }(c) & = \frac{ k_{ \textnormal{in} } }{  \gamma_0 \gamma_1 + c \gamma_0k_{ \textnormal{off} }  + c \gamma_1 k_{ \textnormal{on} }   }
[
\begin{array}{cc}
\pi_0 & \pi_1
\end{array}
]
\left[
\begin{array}{cc}
 \gamma_1 + c k_{ \textnormal{off} }  &   c k_{ \textnormal{on} }\\
  c k_{ \textnormal{off} } &  \gamma_0 + c k_{ \textnormal{on} }
\end{array}
\right]
\left[
\begin{array}{c}
1 \\
1
\end{array}
\right] \notag \\
& = \frac{ k_{ \textnormal{in} } }{  ( \gamma_0 \gamma_1 + c \gamma_0k_{ \textnormal{off} }  + c \gamma_1 k_{ \textnormal{on} }  )  (  k_{ \textnormal{on} }  +   k_{ \textnormal{off} } )}
[
\begin{array}{cc}
k_{ \textnormal{off} }   & k_{ \textnormal{on} }
\end{array}
]
\left[
\begin{array}{c}
 \gamma_1 +    c ( k_{ \textnormal{on} }  +   k_{ \textnormal{off} } ) \\
 \gamma_0 +    c ( k_{ \textnormal{on} }  +   k_{ \textnormal{off} } )
\end{array}
\right] \notag \\
& = k_{ \textnormal{in} } \left[ \frac{   \gamma_1 k_{ \textnormal{off} }  + \gamma_0 k_{ \textnormal{on} }  +    c ( k_{ \textnormal{on} }  +   k_{ \textnormal{off} } )^2  }{  ( \gamma_0 \gamma_1 + c \gamma_0k_{ \textnormal{off} }  + c \gamma_1 k_{ \textnormal{on} }  )  (  k_{ \textnormal{on} }  +   k_{ \textnormal{off} } )} \right].
\end{align}
\end{linenomath*}

It is interesting to point our the formal similarity of \eqref{main_formula_two_enz} with the formula for the mean transfer time of a relaxation process whose rate is modeled by a two-state CTMC \cite{goychuk2005rate}.
From \eqref{main_formula_two_enz} it can be readily seen that the steady-state substrate means in the static and deterministic cases are given by
\begin{linenomath*}
\begin{align*}
m_{ \textnormal{eq} }^{ \textnormal{(det)} } &= \lim_{c \to \infty}  m_{ \textnormal{eq} }(c) = k_{ \textnormal{in} } \left[ \frac{   k_{ \textnormal{on} }  +   k_{ \textnormal{off} }   }{   \gamma_0k_{ \textnormal{off} } +  \gamma_1 k_{ \textnormal{on} }  } \right]  \\
\textnormal{and} \quad  m_{ \textnormal{eq} }^{ \textnormal{(static)} } &= \lim_{c \to 0} m_{ \textnormal{eq} }(c) = k_{ \textnormal{in} } \left[ \frac{  \gamma_1 k_{ \textnormal{off} }  + \gamma_0 k_{ \textnormal{on} }    }{ \gamma_0 \gamma_1  (  k_{ \textnormal{on} }  +   k_{ \textnormal{off} } ) } \right].
\end{align*}
\end{linenomath*}
Hence we can compute the maximum relative amplification factor $\rho_{ \textnormal{max} }$ \eqref{defn_rho_max} as
\begin{linenomath*}
\begin{align}
\label{defn_rho_max_twostateenz}
\rho_{ \textnormal{max} } &= \frac{m_{ \textnormal{eq} }^{ \textnormal{(static)} }  }{m_{ \textnormal{eq} }^{ \textnormal{(det)} } } - 1 =  \frac{ (  \gamma_0k_{ \textnormal{off} } +  \gamma_1 k_{ \textnormal{on} })  (  \gamma_1 k_{ \textnormal{off} }  + \gamma_0 k_{ \textnormal{on} }   ) }{  (  k_{ \textnormal{on} }  +   k_{ \textnormal{off} } )^2  \gamma_0 \gamma_1} -1  = \frac{ k_{ \textnormal{on} }   k_{ \textnormal{off} }  ( \gamma_0 - \gamma_1 )^2  }{ (  k_{ \textnormal{on} }  +   k_{ \textnormal{off} } )^2  \gamma_0 \gamma_1  }.
\end{align}
\end{linenomath*}
Recall the formula for $\textnormal{Var}_\pi( \gamma )$ from \eqref{twostateenzyme:mean_and_var} and oberve that $\rho_{ \textnormal{max} }$ can be expressed as
\begin{linenomath*}
\begin{align*}
\rho_{ \textnormal{max} }  = \frac{ \textnormal{Var}_\pi(  \gamma ) }{\gamma_0  \gamma_1},
\end{align*}
\end{linenomath*}
which reinforces the point we made in Section \ref{sec:approxformula} that $\rho_{ \textnormal{max} }$ serves as a proxy for the variance of the stationary distribution.

The Dirichlet form $\Theta$ \eqref{thetaisdif} for this CTMC is given by
\begin{linenomath*}
\begin{align*}
\Theta = - [
\begin{array}{cc}
\pi_0 & \pi_1
\end{array}
]
\left[
\begin{array}{cc}
\frac{1}{\gamma_0} &0 \\
0 & \frac{1}{\gamma_1}
\end{array}
\right]
Q
\left[
\begin{array}{cc}
\frac{1}{\gamma_0} &0 \\
0 & \frac{1}{\gamma_1}
\end{array}
\right]
\left[
\begin{array}{c}
1 \\
1
\end{array}
\right] = \left( \frac{\gamma_0  -  \gamma_1}{  \gamma_0 \gamma_1 } \right)^2 \frac{ k_{ \textnormal{on} }   k_{ \textnormal{off} }   }{   k_{ \textnormal{on} }  +   k_{ \textnormal{off} }  }.
\end{align*}
\end{linenomath*}
This yields the following formula for the normalized Dirichlet form $\theta$ \eqref{thetaisdif2}
\begin{linenomath*}
\begin{align*}
\theta = \frac{\Theta \E_\pi(\gamma)}{ \rho_{ \textnormal{max} } } =   \frac{k_{ \textnormal{off} } \gamma_0 +  k_{ \textnormal{on} }  \gamma_1  }{ \gamma_0  \gamma_1 }
\end{align*}
\end{linenomath*}
which determines the shape of our approximate formula $\hat{m}_{ \textnormal{eq} }(c)$ \eqref{mceqapprox} for the steady-state substrate mean. It is straightforward to check that for this example, this approximate formula is \emph{exact} because the expression \eqref{main_formula_two_enz} for $m_{ \textnormal{eq} }(c)$ can be written as
\begin{linenomath*}
\begin{align*}
m_{ \textnormal{eq} }(c) = m_{ \textnormal{eq} }^{ \textnormal{(det)} }  + \left( \frac{  m_{ \textnormal{eq} }^{ \textnormal{(static)} }  - m_{ \textnormal{eq} }^{ \textnormal{(det)} }  }{1+\theta c} \right).
\end{align*}
\end{linenomath*}
Note that $\theta$ is a measure of the mixing strength of the enzymatic kinetics and so it is not surprising that it increases linearly with the transition rates $k_{ \textnormal{on} }$ and $k_{ \textnormal{off} }$.

Define the relative amplification factor by \eqref{defn_stochampl}. It can be exactly expressed as
\begin{linenomath*}
\begin{align*}
\rho(c) = \frac{  \rho_{ \textnormal{max} } }{1+\theta c} = \frac{ k_{ \textnormal{on} }   k_{ \textnormal{off} }  ( \gamma_0 - \gamma_1 )^2  }{ (  k_{ \textnormal{on} }  +   k_{ \textnormal{off} } )^2  ( \gamma_0 \gamma_1 + c \gamma_0k_{ \textnormal{off} }  + c \gamma_1 k_{ \textnormal{on} }  )   }. 
\end{align*}
\end{linenomath*}
We now consider the situation when the degradation rate induced by the enzyme {\bf E} in the low-activity state (``0") is negligible. In this case $\gamma_0 \approx 0$ and $\rho(c)$ simplifies to
\begin{linenomath*}
\begin{align*}
\rho(c) \approx \left( \frac{   k_{ \textnormal{off} }  \gamma_1  }{ (  k_{ \textnormal{on} }  +   k_{ \textnormal{off} } )^2    } \right) \frac{1}{c} = \left( \frac{   \pi_0 \gamma_1  }{ k_{ \textnormal{on} }  +   k_{ \textnormal{off} }  } \right) \frac{1}{c},
\end{align*}
\end{linenomath*}
which shows that the relative amplification factor is proportional to $1/c$ and the proportionality constant is simply the product of the proportion of time ($\pi_0$) the enzyme spends in the low-activity state, the degradation rate ($\gamma_1$) at the high-activity state and the reciprocal of the sum of transition rates $k_{ \textnormal{on} }$ and $k_{ \textnormal{off} }$. In particular as $c$ approaches $0$, the relative amplification factor $\rho(c)$ can be enormous, thereby indicating that such a switching enzyme ${\bf E}$ can exploit its fluctuations to function as a biological amplifier with a very high gain.

\subsection{Stochastic focusing Network} \label{example:sf}
In this section we apply our results to the famous stochastic focusing network given in \cite{Paulsson00}. This network involves three species: substrate {\bf S}, product {\bf P} and enzyme {\bf E}\footnote{In \cite{Paulsson00}, {\bf S} was called {\bf I} and {\bf E} was called {\bf S}. We have changed the notation to ensure consistency with the notation in this paper.}. The molecules of substrate {\bf S} are produced constitutively at rate $k_{ \textnormal{in} }$ and converted into product {\bf P} through a first-order reaction with rate constant $k_p$. Both subtrate and product molecules degrade spontaneously at rates $k_a e$ and $\delta_p$ respectively, where $e$ denotes the current state or abundance level of enzyme {\bf E}. The schematic representation of these reactions is as follows
\begin{linenomath*}
\begin{align}
\label{defn:reac_net_sf}
\emptyset   \xrightleftharpoons[  k_a  e  ]{k_{ \textnormal{in} } } {\bf S}  \stackrel{k_p}{ \longrightarrow }   {\bf P}    \stackrel{ \delta_p }{ \longrightarrow } \emptyset
\end{align}
\end{linenomath*}
The enzymatic dynamics in this example is given by the Markovian birth-death process with birth-rate $k_s$ and death-rate $k_d$:
\begin{linenomath*}
\begin{align}
\label{defn_enzy_kin_sf}
\emptyset  \stackrel{k_s}{ \longrightarrow }   {\bf E}    \stackrel{ k_d }{ \longrightarrow } \emptyset.
\end{align}
\end{linenomath*}
This process evolves on state-space $\N_0$, which is the set of all nonnegative integers and its unique stationary distribution is Poisson with mean $k_s/k_d$. We assume that the initial enzymatic state is a random variable with this stationary distribution.

Multiplying the rate constants $k_s$ and $k_d$ by $c$, we obtain enzymatic kinetics whose speed relative to the substrate is $c$. Let $m_{ \textnormal{eq} }^{({\bf S})}(c)$ and $m_{ \textnormal{eq} }^{({\bf P})}(c)$ denote the steady-state means of substrate and product respectively when the relative enzyme speed is $c$. From the \emph{first-order moment equations} for the network \eqref{defn:reac_net_sf} one can easily show (see Supplementary Material) that for any $c \geq 0$
\begin{linenomath*}
\begin{align}
\label{simpl_reln_sf}
m_{ \textnormal{eq} }^{({\bf P})}(c) = \frac{k_p}{ \delta_p } m_{ \textnormal{eq} }^{({\bf S})}(c).
\end{align}
\end{linenomath*}
To study the amplification of steady-state means due to enzymatic fluctuations we use the relative stochastic amplification factor defined by \eqref{defn_stochampl}. Due to the linear relationship \eqref{simpl_reln_sf} between $m_{ \textnormal{eq} }^{({\bf S})}(c)$ and $m_{ \textnormal{eq} }^{({\bf P})}(c)$, these amplification factors are same for both product and subtrate. Therefore we can understand the amplification phenomenon by replacing network \eqref{defn:reac_net_sf} with our simplified scheme (see Figure \ref{fig:model}), where the degradation rate at time $t$ is given by
\begin{linenomath*}
\begin{align*}
\gamma_c(t) = k_p + k_a E_c(t),
\end{align*}
\end{linenomath*}
and $E_c(t)$ denotes the state at time $t$ of enzymatic kinetics with relative speed $c$. Let $\gamma_i = k_p + k_a i$ for each $i=0,1,\dots$. Note that $ ( \gamma_c(t) )_{t  \geq 0} $ is a CTMC with state-space
\begin{linenomath*}
\begin{align*}
\Gamma = k_p + k_a \N_0 = \{ \gamma_1,\gamma_2,\dots\}
\end{align*}
\end{linenomath*}
and stationary distribution\footnote{This stationary distribution is obtained by applying the linear change of variables $\gamma = k_p + k_a e $ on the Poisson distribution with mean $k_s/k_d$.}
\begin{linenomath*}
\begin{align}
\label{def_stpi:sf}
\pi(\gamma_i) = \frac{e^{ - \frac{k_s}{k_d}  }}{i!} \left(  \frac{k_s}{k_d}  \right)^i \qquad \textnormal{for any } \qquad \gamma_i = k_p + k_a i \in \Gamma.
\end{align}
\end{linenomath*}
This CTMC transitions from state $\gamma_i$ to state $\gamma_{i+1}$ at rate $c k_s$ and from state $\gamma_i$ to state $\gamma_{i-1}$ at rate $c i k_d$. In other words, the generator for this CTMC is given by
\begin{linenomath*}
\begin{align}
\label{gen_sf}
\mathbb{Q}_c f( \gamma_i ) = c  k_s (f (\gamma_{i+1}) - f( \gamma_i ) ) +c i k_d (f (\gamma_{i-1}) - f( \gamma_i ) )
\end{align}
\end{linenomath*}
for any bounded function $f : \Gamma \to \R$.

In the rest of this section, we denote the steady-state substrate mean by $m_{ \textnormal{eq} }(c)$ instead of $m_{ \textnormal{eq} }^{({\bf S})}(c)$. Since the state-space $\Gamma$ is not finite, we cannot use the results from Section \ref{sec:mainanalysisctmc} to compute $m_{ \textnormal{eq} }(c)$. However we can easily simulate the paths of process $( \gamma_c(t) )_{t  \geq 0}$ with Gillespie's Algorithm \cite{GP}, and obtain samples of the random variable $\tau_c$ defined by \eqref{defn:tau}. The corresponding sample mean then serves as an estimator for $m_{ \textnormal{eq} }(c)$ (see part (C) of Proposition \ref{prop_general}). Note that the steady-state substrate mean in the absence of enzymatic fluctations is simply given by
\begin{linenomath*}
\begin{align*}
m_{ \textnormal{eq} }^{ \textnormal{(det)} } =\frac{k_{  \textnormal{in} }}{ k_p + k_a  \left(\frac{k_s}{k_d} \right)  }.
\end{align*}
\end{linenomath*}
Dividing $m_{ \textnormal{eq} }(c)$ by $m_{ \textnormal{eq} }^{ \textnormal{(det)} }$ and subtracting $1$, we obtain an estimate for the relative stochastic amplification factor $\rho(c)$ (see \eqref{defn_stochampl}). This factor only depends on four rate constants $k_p, k_a,k_s$ and $k_d$ which we now set as
\begin{linenomath*}
\begin{align}
\label{rate_constants_sf}
k_p = 0.35, \quad k_a = 0.25, \quad k_d = 1 \quad \textnormal{and} \quad k_s = 1.
\end{align}
\end{linenomath*}
We estimate $\rho(c)$ for several values of $c$ in the interval $(0,20)$ and plot these estimates in Figure \ref{fig1:sf}. For each value of $c$, $\rho(c)$ was estimated using $10^5$ samples of $\tau_c$ and the resulting standard error\footnote{The standard error is simply the standard deviation of the distribution of the sample mean.} is also displayed in Figure \ref{fig1:sf}. In Section \ref{sec:approxformula} we develop an approximate expression $\hat{ \rho}(c)$ \eqref{main_ampl_formula} for the relative amplification factor which is likely to hold even though the state-space $\Gamma$ is not finite. Using the stationary distribution $\pi$ \eqref{def_stpi:sf} and the generator $\mathbb{Q}_c$ (with $c=1$) we estimate the maximum amplification factor $\rho_{ \textnormal{max} }$ \eqref{defn_rho_max} and the normalized Dirichlet form $\theta$ \eqref{thetaisdif2} as $\rho_{ \textnormal{max} } = 0.1703$ and $\theta  = 2.1988$ respectively. With these values we evaluate the map $c \mapsto \hat{ \rho}(c)$ and plot it in the interval $(0,20)$ in Figure \ref{fig1:sf}. The close agreement between the estimated and the approximate values of the relative amplification factor can be easily seen. In Figure \ref{fig1:sf} we also indicate the threshold speed $c_\epsilon$ \eqref{cthreshold} for the $1\%$ threshold level (i.e. $\epsilon = 0.01$). This threshold speed is $c_\epsilon = 7.2903$ which indicates that if $c < 7.2903$ then enzymatic fluctuations will amplify the steady-state substrate mean by more than $1\%$ in comparison to the deterministic case. In other words, the relative error in assuming that the enzyme activity is deterministic \emph{exceeds} $1\%$ if $c <7.2903$.

\begin{figure}[h]
\begin{center}
    \includegraphics[width=0.7\textwidth]{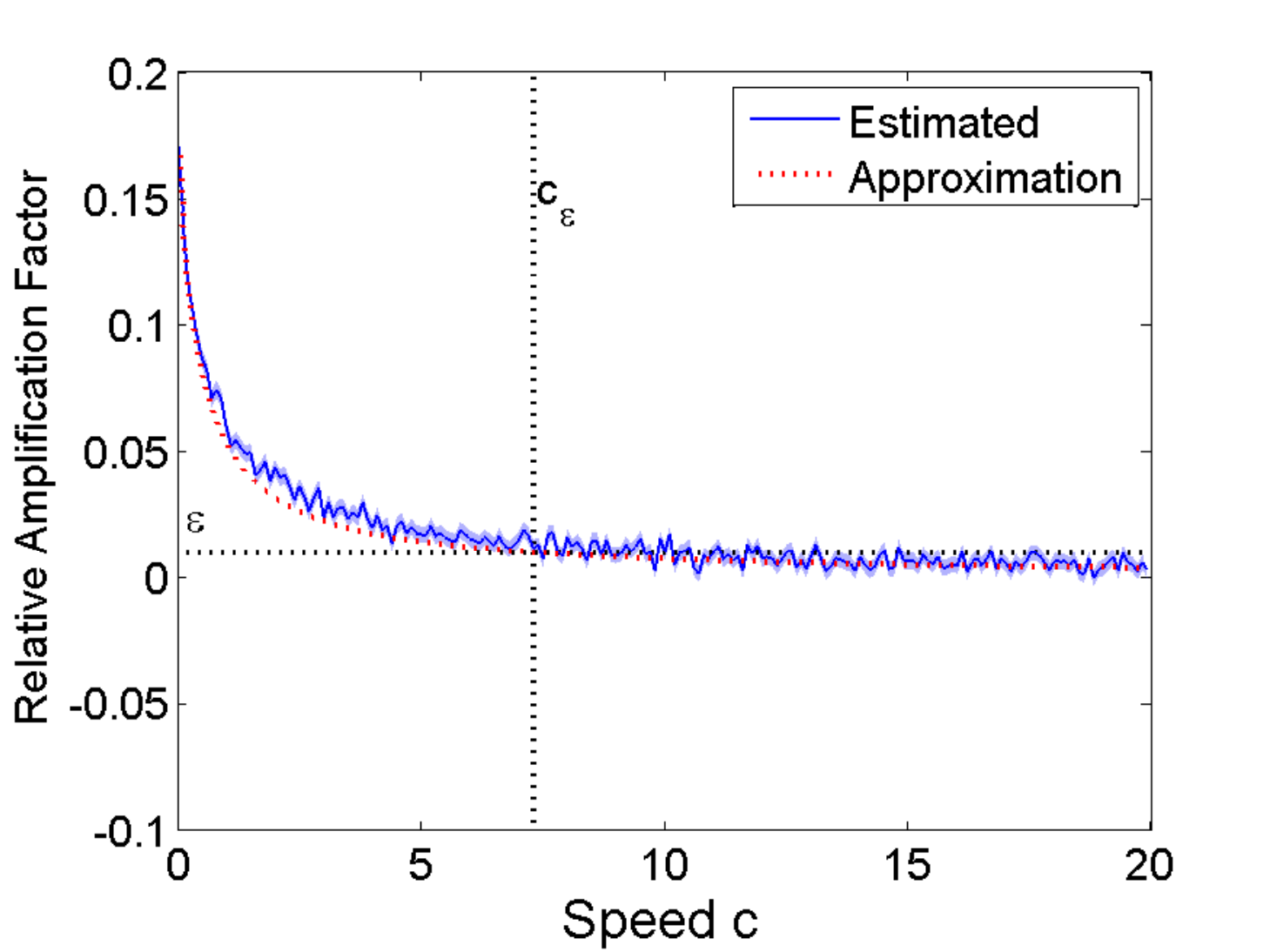}
\end{center}
    \caption{{\bf Stochastic focusing network:} Comparison of the estimated values of the exact relative amplification factor $\rho(c)$ \eqref{defn_stochampl} with the approximate factor $\hat{\rho}(c)$ obtained from formula \eqref{main_ampl_formula}. The threshold speed $c_\epsilon$ for the $1\%$ tolerance level ($\epsilon = 0.01$) is also marked. The estimated values are obtained using the estimator based on formula \eqref{eqnetauc} with $10^5$ samples. The shaded region represents the symmetric one standard deviation interval around the mean.}
    \label{fig1:sf}
\end{figure}

We now explore the effects of changing the levels of noise in the enzymatic activity, on the relative amplification factor for the steady-state substrate mean. This noise can be measured using the \emph{coefficient of variation} (CV)\footnote{The coefficient of variation of a probability distribution is its standard deviation divided by its mean. It measures the dispersion of a distribution relative to its mean. } of the stationary distribution for the enzyme abundance. Since this distribution is Poisson with mean $k_s/k_d$, the CV is $(\sqrt{k_d}/\sqrt{k_s})$ which shows that for a fixed $k_d$, we can decrease the relative noise level by simply increasing $k_s$. With this in mind we repeat the above computations (see Figure \ref{fig1:sf}) for three additional values of $k_s$: 5, 10 and 20, and the results are provided in Figure \ref{fig2:sf}. For each value of $k_s$, the corresponding estimates for the maximum relative amplification factor $\rho_{ \textnormal{max} }$, the normalized Dirichlet form $\theta$ and the threshold speed $c_\epsilon$ (for $\epsilon = 0.01$) are given in Table \ref{tab:sf}.

\begin{table}[h]
\centering
\begin{tabular}{|c|c|c|c|} \hline
$k_s$ & $\rho_{ \textnormal{max} }$ & $\theta$ & $c_\epsilon$ \\   \hline
1 & 0.1703 & 2.1988 & 7.2903 \\
5 & 0.1543  & 1.6907 & 8.5349  \\
10 & 0.0930 &0.6782  &12.2383 \\
20 & 0.0494 & 0.2550 & 15.4510 \\  \hline
\end{tabular}
  \caption{Estimates for $\rho_{ \textnormal{max} }$, $\theta$ and $c_\epsilon$ for various values of $k_s$}
\label{tab:sf}
\end{table}

\begin{figure}[h]
\begin{center}
    \includegraphics[width=\textwidth]{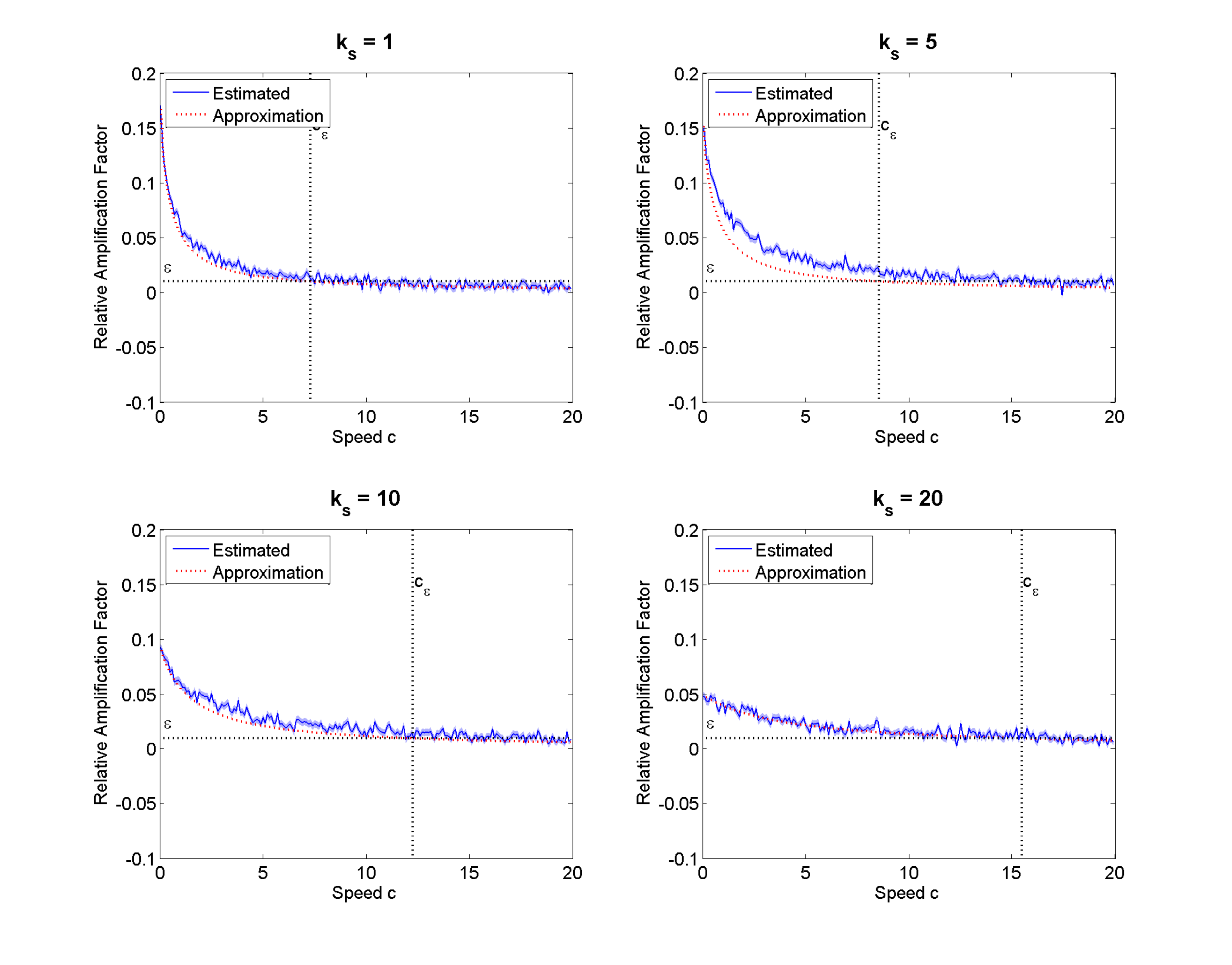}
\end{center}
    \caption{{\bf Stochastic focusing network:} Comparison of the relative amplification factors for various values of $k_s$. As $k_s$ gets larger the level of noise in the enzymatic dynamics decreases. Consequently the maximum value $\rho_{ \textnormal{max} }$ (achieved at speed $c = 0$) declines but the values of the normalized Dirichlet form $\theta$ also declines. This results in lesser convexity of the map $c \mapsto \rho(c)$ which can be seen from the increasing values of the threshold speeds $c_\epsilon$ for $\epsilon = 0.01$.}
    \label{fig2:sf}
\end{figure}

Recall the discussion at the end of Section \ref{sec:approxformula} on the effects of noise in the enzymatic dynamics. From Table \ref{tab:sf} it is clear that as expected, decreasing noise (or increasing $k_s$) results in the decline of both $\rho_{ \textnormal{max} }$ and $\theta$. These parameters influence the threshold speed $c_\epsilon$ (see \eqref{cthreshold}) in opposite ways, but their overall effect is to \emph{increase} $c_\epsilon$, indicating that as the noise levels go down, the relative enzyme speed needs to be higher and higher for the assumption of deterministic enzymatic activity to be acceptable. In other words, even though noise in enzyme activity causes the stochastic amplification effect it also helps in eliminating it.

\section{Discussion} \label{sec:disc}
We examined the mathematical properties of a system consisting of a substrate that is degraded through an enzyme with stochastically fluctuating activity levels. Our analysis focused on the effect of enzymatic fluctuations on the mean substrate abundance and its deviations from the deterministic model predictions. It should be pointed out that even if the substrate inflow rate is assumed to be an independent stationary stochastic process with mean $k_{ \textnormal{in} }$, our results will not be affected.

Whereas a stochastically varying production rate would leave the mean substrate level unaffected and equal to that of the deterministic model, fluctuations in the removal rate of the substrate result in a system that behaves very differently in the stochastic and deterministic regimes due to the product term in the degradation rate of $S$. Our formulas help quantify this discrepancy and study its behavior as the speed of enzymatic fluctuations varies from zero to infinity. They also provide an interesting connection between the amplification effect and the mixing properties of the Markov process describing the enzymatic activity fluctuations, which allow us to determine the speed above which this amplification becomes negligible for a given system parametrization. Note that the study of such systems through the use of approximate stochastic models such as the Linear Noise Approximation \cite{Elf03} is particularly challenging, since these methods typically fail to capture the very strong negative correlations between enzyme activity and substrate that can arise at slow enzyme fluctuations (see also discussion in \cite{Milias15}). On the contrary, the results presented here are valid under much milder simplifying assumptions, and can thus accurately reveal the magnitude of the discrepancy between stochastic and deterministic descriptions of the system.

Given the prevalence of enzymatic interactions in cell biology, stochastic fluctuations in enzyme activity and/or abundance are expected to play a large role in shaping the mean intracellular abundances of substrates \cite{Iversen14}, which could potentially also deviate significantly from the deterministically predicted amounts. Since many enzymes are allosterically regulated \cite{Nussinov16} by their products, substrates or other small signaling molecules, it would be very interesting to also study the effects of this regulation on the statistics of substrates and products, and examine potential noise reduction \cite{Milias15} or signal amplification strategies. As the sensitivity of single-molecule enzymology experimental techniques increases, it may soon be possible to study the phenomena described theoretically in this work within living cells.

\clearpage
\appendix
\setcounter{equation}{0}
\setcounter{figure}{0}
\setcounter{table}{0}
\setcounter{section}{0}
\renewcommand\thesection{S\arabic{section}}

\makeatletter
\renewcommand{\theequation}{S\arabic{equation}}
\renewcommand{\thefigure}{S\arabic{figure}}

\begin{center}
\title{\Large{ \bf Supplementary Material}}
\end{center}

\section{The model} \label{Ssec:model}

In this paper we consider a system where the substrate ${\bf S}$ enters at a constant rate $k_{ \textnormal{in} }$ and is degraded at a rate that depends on the activity state or abundance level of an enzyme ${\bf E}$. This activity state is assumed to fluctuate in time $t$ according to a continuous-time Markov chain (CTMC) $(\gamma(t) )_{t \geq 0 }$ over a finite state-space $\Gamma = \{ \gamma_1,\dots,\gamma_n \}$. The system can be written as
\begin{align}
\label{Sbasic_rnt}
\emptyset  \stackrel{ k_{ \textnormal{in} } }{\longrightarrow} {\bf S}  \stackrel{  \gamma(t) }{\longrightarrow} \emptyset.
\end{align}

The CTMC $(\gamma(t) )_{t \geq 0 }$ is described by its $n \times n$ transition rate matrix $Q=[q_{ij}]$ (see \cite{Norris98}). For any distinct $i,j \in \{1,2,\dots,n\}$, $q_{ij} \geq 0$ denotes the rate at which the process leaves state $\gamma_i$ and enters state $\gamma_j$. The diagonal entries of $Q$ are given by $q_{ii} = -\sum_{j \neq i}q_{ij}$. From now on we assume that the rate matrix $Q$ is irreducible\footnote{A matrix $Q$ is called \emph{irreducible} if there does not exist a permutation matrix $P$ such that the matrix $P Q P^{-1}$ is block upper-triangular.} which implies that there exists a unique stationary distribution $\pi = (\pi_1,\dots,\pi_n) \in \R^n_+$ satisfying
\begin{align*}
Q {\bf 1} = {\bf 0},  \qquad  \pi^T Q  = {\bf 0}^T  \qquad \textnormal{and} \qquad \pi^T {\bf 1} = {\bf 1},
\end{align*}
where ${\bf 0}$ and ${\bf 1}$ denote the $n \times 1$ vectors of all zeroes and ones respectively. Since the state-space is finite and the transition rate matrix $Q$ is irreducible, the CTMC $( \gamma (t) )_{t  \geq 0 }$ is \emph{ergodic} which means that the probability distribution of $\gamma(t)$ converges to the stationary distribution $\pi$ as $t \to \infty$. As we are interested in the steady-state limit, without loss of generality we can assume that the initial state $\gamma(0)$ is distributed according to $\pi$, i.e. $\P( \gamma(0) = \gamma_i )$ for each $i=1,\dots,n$. This ensures that the process $(\gamma(t) )_{t \geq 0 }$ is a \emph{stationary} stochastic process whose finite-dimensional distributions are invariant under time-shifts. This means that for any finite collection of time-points $t_1,t_2,\dots,t_n$ the joint distribution of the random vector $(\gamma(t_1+s), \dots, \gamma(t_n+s) )$ remains the same for all $s \geq 0$. This also implies that various statistical properties of this process do not depend on time. In particular its mean $\E( \gamma(t) )$ is equal to
\begin{align}
\label{Sstationarymean:gamma}
\E( \gamma(t) ) = \E_\pi(\gamma)=  \sum_{i=1}^n \gamma_i \pi_i \qquad \textnormal{for all} \quad t \geq 0,
\end{align}
where $\gamma$ is a $\Gamma$-valued random variable with probability distribution $\pi$ and $\E_\pi(\cdot)$ denotes the expectation w.r.t. this distribution.

In what follows, we need to view enzymatic dynamics at the timescale of substrate kinetics. For this we define a family of processes $( \gamma_c(t) )_{t  \geq 0 }$ parameterised by the ``relative speed" parameter $c$ as
\begin{align}
\label{Sdefngammac}
\gamma_c(t) = \gamma(c t) \qquad \textnormal{for all} \quad t \geq 0.
\end{align}
 Like $( \gamma(t) )_{t \geq 0}$, the process $( \gamma_c(t) )_{t \geq 0}$ is also a CTMC over state-space $\Gamma = \{ \gamma_1,\dots,\gamma_n \}$ with transition rate matrix $Q_c= cQ$ and initial distribution $\pi$. Since $( \gamma(t) )_{t \geq 0}$ is stationary, this process is also stationary with the same mean given by $\E_\pi(\gamma) = \E( \gamma_c(t) ) $ for all times $t \geq 0$. Replacing $( \gamma(t) )_{t \geq 0}$ by $( \gamma_c(t) )_{t  \geq 0 }$ in \eqref{Sbasic_rnt}, we will study how the steady-state mean of substrate abundance depends on the fluctuation speed $c$.

Given a sample path of the enzyme dynamics $( \gamma_c(t) )_{t  \geq 0 }$ with relative speed $c$, we regard the dynamics of substrate molecular counts as a jump Markov chain $( S_c(t))_{t \geq 0}$ over the set of nonnegative integers $\N_0= \{0,1,2,\dots\}$. This Markov chain can be written in the random time change representation \cite{EK} as
\begin{align}
\label{Ssub_rtc}
S_c(t) = S_c(0) + Y_1(  k_{ \textnormal{in} } t  ) - Y_2\left( \int_0^t \gamma_c(u) S_c(u)du \right),
\end{align}
where $Y_1$ and $Y_2$ are independent, unit rate Poisson processes. Here the Poisson processes $Y_1$ and $Y_2$ capture the intermittency in the firing of production and degradation reactions. This intermittency becomes unimportant if the substrate is present in high copy-numbers \cite{KurtzLLn1} and in this case one can regard $( S_c(t))_{t \geq 0}$ as the dynamics of substrate \emph{concentration}\footnote{The concentration of any species is its copy-number divided by the system volume.}, specified by the following ODE
\begin{align}
\label{Ssub_ode}
\frac{ d S_c(t)}{d t} = k_{ \textnormal{in} } -  \gamma_c(t) S_c(t).
\end{align}
Let $m_c(t) = \E(S_c(t))$ for each $t \geq 0$, where $(S_c(t))_{t  \geq 0 }$ evolves according to either \eqref{Ssub_rtc} or \eqref{Ssub_ode}. Our goal in this paper is to understand the role of fluctuations in the catalytic activity of enzyme ${\bf E}$ in determining the steady-state value of the mean
\begin{align}
\label{Ssteadystatemean}
m_\textnormal{eq}(c) = \lim_{t \to \infty} m_c(t).
\end{align}

\section{Expressions for $m_{ \textnormal{eq} }(c)$: The general case } \label{Ssec:mainanalysis} 
In this section we prove Proposition 2.1 in the main paper. For convenience, we restate this proposition below.
\begin{proposition}
\label{Sprop_general_supp}
Suppose $ ( \gamma(t) )_{t \geq 0 }$ is a real-valued stationary stochastic process with stationary distribution $\pi$ and state-space $\Gamma$ satisfying
\begin{align}
\label{Sgamma_inf}
\inf\{x : x \in \Gamma \}  \geq \epsilon
\end{align}
for some $\epsilon >0$. Let $ ( \gamma_c(t) )_{t \geq 0}$ be the speed $c$ version of this process given by \eqref{Sdefngammac} and define the substrate dynamics $( S_c(t) )_{t \geq 0 }$ either by \eqref{Ssub_rtc} or by \eqref{Ssub_ode}. Let $m_c(t) =\E (S_c(t))$ and let the steady-state limit $m_{ \textnormal{eq} }(c)$ be given by \eqref{Ssteadystatemean}. Then we have the following:
\begin{itemize}
\item[(A)] The value $m_{ \textnormal{eq} }(c)$ is well-defined (i.e. the limit in \eqref{Ssteadystatemean} exists) and it is given by
\begin{align}
\label{Smceqform1}
m_{ \textnormal{eq} }(c) = \lim_{t \to \infty}m_c(t)=   k_{ \textnormal{in} } \int_{0}^\infty   \E \left( e^{ -  \int_0^{s} \gamma_c(u)du }    \right) ds.
\end{align}
\item[(B)] Let $\tau_c$ is the random variable defined by
\begin{align}
\label{Sdefn:tau}
\tau_c = \inf\left\{ t \geq 0 :  \int_0^t \gamma_c(s)ds = - \ln u \right\},
\end{align}
where $u$ is an independent random variable with the uniform distribution on $[0,1]$. Then we have
\begin{align}
\label{Seqnetauc}
m_{ \textnormal{eq} } (c) = k_{ \textnormal{in} } \E( \tau_c ).
\end{align}
\item[(C)] The limits below are satisfied as $c \to \infty$ and $c \to 0$ respectively:
\begin{align}
\label{Smdeterministic}
\lim_{c \to \infty} m_{ \textnormal{eq} }(c) &=  k_{ \textnormal{in} } \int_{0}^\infty  e^{ -s \E_\pi( \gamma ) } ds = \frac{ k_{ \textnormal{in} } }{ \E_\pi( \gamma ) } := m^{ \textnormal{(det)} }_{ \textnormal{eq} } \\
\label{Smstatic}
\textnormal{and} \qquad  \lim_{c \to 0} m_{ \textnormal{eq} }(c) &=  k_{ \textnormal{in} }\E\left(   \int_{0}^\infty  e^{ -s  \gamma(0) } ds \right)=  k_{ \textnormal{in} } \E \left( \frac{1}{\gamma(0)} \right) = k_{ \textnormal{in} } \E_\pi\left( \frac{1}{\gamma} \right):= m^{ \textnormal{(static)} }_{ \textnormal{eq} }.
\end{align}
\end{itemize}
\end{proposition}
\begin{proof}
Let $\{ \mathcal{F}_t\}$ be the filtration generated by the process $ (\gamma_c(t) )_{ t \geq 0}$ and let $\mathcal{F}_\infty = \lim_{ t \to \infty} \mathcal{F}_t$ be its limiting value. Given the information $\mathcal{F}_\infty$, the random path of the process $ ( \gamma_c(t) )_{ t \geq 0}$ is completely known. Hence we can formulate the ODE for the conditional first-moment $\E( S_c(t) \vert \mathcal{F}_\infty )$ as follows
\begin{align}
\label{Smainode}
\frac{d}{dt} \E( S_c(t) \vert \mathcal{F}_\infty ) = k_{ \textnormal{in} } - \gamma_c(t) \E( S_c(t) \vert \mathcal{F}_\infty ).
\end{align}
This equation remains unchanged whether we use representation \eqref{Ssub_rtc} or \eqref{Ssub_ode} for the substrate dynamics $(S_c(t))_{t \geq 0}$. Using $\exp( \int_0^t \gamma_c(s)ds )$ as the integrating factor we can write \eqref{Smainode} as
\begin{align*}
\frac{d}{dt} \left(  e^{ \int_0^t \gamma_c(s)ds }  \E( S_c(t) \vert \mathcal{F}_\infty ) \right) = k_{\textnormal{in}}e^{ \int_0^t \gamma_c(s)ds }.
\end{align*}
Finally, integrating both sides w.r.t. time $t$ we obtain
\begin{align*}
 \E( S_c(t) \vert \mathcal{F}_\infty ) =  S_c(0) e^{ -  \int_0^t \gamma_c(s)ds } + k_{\textnormal{in}} \int_{0}^t   e^{ -  \int_s^t \gamma_c(u)du }    ds.
\end{align*}
Taking expectations we get
\begin{align}
\label{Smainodesols}
m_c(t)= \E( S_c(t)  ) &=  S_c(0) \E \left( e^{ -  \int_0^t \gamma_c(s)ds } \right)  +  k_{\textnormal{in}} \E\left( \int_{0}^t  e^{ -  \int_s^t \gamma_c(u)du }  \right)    ds \notag  \\
& = S_c(0)  \E \left( e^{ -  \int_0^t \gamma_c(s)ds } \right) +  k_{\textnormal{in}} \int_{0}^t  \E \left( e^{ -  \int_s^t \gamma_c(u)du }   \right)    ds ,
\end{align}
where the last equation follows from the Fubini's theorem. Since the process $( \gamma_c(t) )_{t  \geq 0 }$ is stationary, the distribution of the random variable $ \int_s^t \gamma_c(u)du$ is same as the distribution of $ \int_0^{t-s} \gamma_c(u)du$. Hence using a simple change of variables we can write
\begin{align*}
 \int_{0}^t   \E \left( e^{ -  \int_s^t \gamma_c(u)du }  \right)  ds  =  \int_{0}^t   \E \left( e^{ -  \int_0^{t-s} \gamma_c(u)du }  \right)  ds =  \int_{0}^t   \E \left( e^{ -  \int_0^{s} \gamma_c(u)du } \right)  ds ,
\end{align*}
which gives us the following formula for $m_c(t)$
\begin{align}
\label{Sformulamt}
m_c(t)=  S_c(0)  \E \left( e^{ -  \int_0^t \gamma_c(s)ds } \right)  +  k_{ \textnormal{in} } \int_{0}^t   \E \left( e^{ -  \int_0^{s} \gamma_c(u)du }  \right)  ds.
\end{align}
Using \eqref{Sgamma_inf} we see that
\begin{align*}
\E( S_c(0)e^{ -  \int_0^t \gamma_c(s)ds } ) \leq \E( S_c(0) )e^{ -  t \epsilon }  \to 0  \qquad \textnormal{as} \qquad t \to \infty,
\end{align*}
and using the dominated convergence theorem we can conclude that
\begin{align*}
\lim_{t \to \infty} \int_{0}^t   \E \left( e^{ -  \int_0^{s} \gamma_c(u)du } \right)   ds  = \int_{0}^{\infty}   \E \left( e^{ -  \int_0^{s} \gamma_c(u)du }  \right)  ds .
\end{align*}
Taking the limit $t \to \infty$ in \eqref{Sformulamt} we obtain the formula \eqref{Smceqform1} for $m_{ \textnormal{eq} }(c)$. This finishes the proof of part (A) of the proposition.

Let $\tau_c$ be the random variable defined by \eqref{Sdefn:tau}. Since it is a continuous random variable with range $[0,\infty)$ we can write
\begin{align}
\label{Sexpandprob}
 \E( \tau_c ) = \int_0^\infty \P( \tau_c > t ) dt.
\end{align}
Let $\{ \mathcal{F}_t \}$ be the filtration as defined before. As $u$ in \eqref{Sdefn:tau} is a uniform $[0,1]$ random variable independent of the process $( \gamma_c(t) )_{t  \geq 0}$ we have
\begin{align*}
\P\left( \tau_c > t  \vert  \mathcal{F}_t   \right) & = \P\left(  \int_0^t \gamma_c(s)ds  < - \ln u  \middle\vert  \mathcal{F}_t   \right)  = \P\left( u < e^{- \int_0^t \gamma_c(s)ds}   \middle\vert  \mathcal{F}_t   \right) = e^{- \int_0^t \gamma_c(s)ds}.
\end{align*}
Taking expectations both sides we get
\begin{align*}
\P( \tau_c > t ) = \E\left( e^{- \int_0^t \gamma_c(s)ds} \right).
\end{align*}
Substituting this in \eqref{Sexpandprob} we obtain
\begin{align*}
 \E( \tau_c ) = \int_0^\infty \P( \tau_c > t ) dt = \int_0^\infty  \E\left( e^{- \int_0^t \gamma_c(s)ds} \right) dt .
\end{align*}
Replacing this integral in \eqref{Smceqform1} by $\E(\tau_c)$ proves formula \eqref{Seqnetauc}. This finishes the proof of part (B) of the proposition. The proof of part (C) is already outlined in the main paper.
\end{proof}

\section{Expressions for $m_{ \textnormal{eq} }(c)$: The finite CTMC case} \label{Ssec:mainanalysisctmc}

In this section we prove Theorem 2.2 in the main paper. For convenience, we restate this theorem below.
\begin{theorem}
\label{Stheorem_general_supp}
Suppose $ ( \gamma(t) )_{t \geq 0 }$ is a stationary CTMC with transition rate matrix $Q$, stationary distribution $\pi$ and state-space $\Gamma = \{ \gamma_1,\dots,\gamma_n \}$ (see Section \ref{Ssec:model}). Let $ ( \gamma_c(t) )_{t \geq 0}$ be the speed $c$ version of this process given by \eqref{Sdefngammac} and define the substrate dynamics $( S_c(t) )_{t \geq 0 }$ either by \eqref{Ssub_rtc} or by \eqref{Ssub_ode}. Let the steady-state substrate mean $m_{ \textnormal{eq} }(c)$ be given by \eqref{Ssteadystatemean} and the diagonal matrix $D$ be defined by
\begin{align}
\label{SdiagmatrixD}
D = \textnormal{Diag}(\gamma_1,\dots,\gamma_n).
\end{align}
Then we have the following:
\begin{itemize}
\item[(A)] The matrix $(D -c Q)$ is invertible and $m_{ \textnormal{eq} }(c)$ can be expressed as
\begin{align}
\label{Smceqform2}
m_{ \textnormal{eq} }(c) = k_{ \textnormal{in} } \left[ \pi^T (D -c Q)^{-1} {\bf 1} \right].
\end{align}
\item[(B)] Suppose the matrix $\tilde{Q} = D^{-1} Q$ is diagonalizable and let $\lambda_1,\dots,\lambda_n$ be its eigenvalues. For each $i=1,\dots,n$ define $\alpha_i$ by
\begin{align}
\label{Sdefn_weights}
\alpha_i =  \langle \pi, u_i \rangle  \langle w_i,  D^{-1}{\bf 1}  \rangle = (\pi^T u_i) ( w^T_i  D^{-1}{\bf 1}  )  \qquad \textnormal{ for each }\qquad i=1,\dots,n.
\end{align}
Then $m_{ \textnormal{eq} }(c)$ can be expressed as
\begin{align}
\label{Smainmceqformula}
m_{ \textnormal{eq} }(c) = k_{ \textnormal{in} } \sum_{i=1}^n  \left( \frac{ \alpha_i }{1 - c \lambda_i } \right).
\end{align}
\item[(C)] The following relation is satisfied for any $c \geq 0$
\begin{align}
\label{Ssandwichprop}
 m^{ \textnormal{(det)} }_{ \textnormal{eq} }  \leq m_{ \textnormal{eq} }(c)  \leq  m^{ \textnormal{(static)} }_{ \textnormal{eq} }.
\end{align}
\end{itemize}
\end{theorem}
\begin{proof}
For each $i = 1,\dots,n$ and $t \geq 0$ define
\begin{align*}
\beta_i(t) = \E\left( \ind_{ \{ \gamma_c(t) = s_i \} } e^{ -\int_0^t \gamma_c(s)ds }  \right).
\end{align*}
Let $ \beta(t)$ denote the vector $\beta(t) = ( \beta_1(t) , \dots, \beta_n(t) )$. Note that
\begin{align*}
{\bf 1}^T  \beta(t) = \sum_{i=1}^m \E\left( \ind_{ \{ \gamma_c(t) = s_i \} } e^{ -\int_0^t \gamma_c(s)ds }  \right) =  \E\left(  \left[ \sum_{i=1}^m \ind_{ \{ \gamma_c(t) = s_i \} } \right] e^{ -\int_0^t \gamma_c(s)ds }  \right) = \E\left(  e^{ -\int_0^t \gamma_c(s)ds }  \right).
\end{align*}
Note that $\beta(0) = \pi$ because for each $i$ we have $\beta_i(0) = \E( \ind_{\{ \gamma_c(0) = s_i \} } ) = \P( \gamma_c(0) =s_i ) = \pi_i$. From Proposition 4.1 in \cite{GuptaEJP} we see that $\beta(t)$ satisfies the following ODE:
\begin{align*}
\frac{d \beta }{dt} = (c Q^T - D) \beta(t).
\end{align*}
Since $\beta(0) = \pi$, the solution to this ODE is
\begin{align*}
\beta(t) = e^{ (c Q^T - D) t } \pi.
\end{align*}
Therefore
\begin{align*}
\E\left(  e^{ -\int_0^t \gamma_c(s)ds }  \right) = {\bf 1}^T  \beta(t) = {\bf 1}^T   e^{ (cQ^T - D) t } \pi.
\end{align*}
On integrating from $t = 0$ to $t = \infty$ we get
\begin{align*}
\int_{0}^{  \infty } \E\left(  e^{ -\int_0^t \gamma_c(s)ds }  \right)dt  =\int_0^\infty {\bf 1}^T \beta(t) dt = {\bf 1}^T \left(  \int_{0}^\infty  e^{ (cQ^T - D) t } dt  \right) \pi =  {\bf 1}^T (D -cQ^T) ^{-1} \pi = \pi^T (D -cQ) ^{-1} {\bf 1}.
\end{align*}
This relation along with \eqref{Smceqform1} proves \eqref{Smceqform2}. Since the matrix $Q$ is irreducible and Metzler with non-positive eigenvalues, the matrix $(D-cQ)$ is invertible and its inverse $(D-cQ)^{-1}$ is a matrix with only nonnegative real entries (see Theorem 2.6 in \cite{Seneta}).

We now provide another proof of \eqref{Smceqform2} using the \emph{Methods of Conditional Moments} (MCM) approach of Hasenauer et al. \cite{Hasenauer13}. We thus define $m_i(t):=\mathbb{E}[S(t)|\gamma(t)=\gamma_i]$ and consider the algebraic equations that describe the system of conditional moments for $m_1,\dots,m_n$:
\begin{equation}\label{MCM_ODEs}\frac{dm_i}{dt}=k\pi_i-\gamma_i\pi_im_i+\sum_{j\neq i}q_{ji}\pi_jm_j-\pi_i m_i\sum_{i\neq j}q_{ij}.\end{equation}
If we define $M_i=\pi_im_i$ and consider the limit as $t\to\infty$ in \eqref{MCM_ODEs}, we get the following system of linear equations:
\begin{equation}k\pi_i-\gamma_iM_i+\sum_{j\neq i}q_{ji}M_j-M_i\sum_{i\neq j}q_{ij}=0.\end{equation}
Letting $M:=\begin{bmatrix}M_1&\dots&M_n\end{bmatrix}$, we get
\[M(\mbox{diag}(\gamma_1,\dots,\gamma_n)-Q)=k\pi\implies M=k\pi(\mbox{diag}(\gamma_1,\dots,\gamma_n)-Q)^{-1}.\]
Therefore,
\[m_{ \textnormal{eq} }(c)=\sum_iM_i=k\pi(D-cQ)^{-1}{\bf 1}.\]

The proof of part (B) is outlined in the main paper. For part (C) note that one of the inequalities $ m^{ \textnormal{(det)} }_{ \textnormal{eq} }  \leq m_{ \textnormal{eq} }(c)$ was already shown in the main paper. Hence it suffices to prove that $m_{ \textnormal{eq} } \leq  m^{ \textnormal{(static)} }_{ \textnormal{eq} }$ which is equivalent to
\begin{align}
\label{Ssuffcondtoshow}
\pi^T  (D -cQ)^{-1} {\bf 1}  \leq  \pi^T D^{-1} {\bf 1}.
\end{align}
Let $\Pi$ be the $m \times m$ diagonal matrix with entries $\pi_1,\dots,\pi_m$. Define another $n  \times n$ matrix by
$$ M =\Pi  (D -c Q)^{-1}  - \Pi  D^{-1}. $$
Since $\Pi  D^{-1}$ is a positive diagonal matrix and $\Pi  (D -Q)^{-1} $ is a componentwise nonnegative matrix, we can conclude that $M$ is also a Metzler matrix. In order to prove \eqref{Ssuffcondtoshow} we just need to show that
\begin{align}
\label{Ssuffcondtoshow2}
 {\bf 1}^T   M  {\bf 1} \leq 0
\end{align}
which is equivalent to proving that
\begin{align}
\label{Ssuffcondtoshow3}
{\bf 1}^T  (M+ M^T ) {\bf 1} \leq 0.
\end{align}
Note that $(M + M^T)$ is a symmetric Metzler matrix. Relation \eqref{Ssuffcondtoshow3} will hold if this matrix is nonpositive definite, which is same as saying that all its eigenvalues have nonpositive real parts. Theorem 2.6 in \cite{Seneta} shows that this matrix is nonpositive definite if we can find a componentwise positive vector $v$ such that
\begin{align}
\label{Scompwisepositivevectorv}
(M + M^T) v = {\bf 0}.
\end{align}
Let $v = (\gamma_1,\dots, \gamma_n)$. Note that since $Q {\bf 1} = {\bf 0}$ we have $(D- cQ) {\bf 1} = D {\bf 1} = v $ and so $(D -cQ)^{-1} v= {\bf 1}$. Therefore
\begin{align}
\label{Spart1}
M v =  \Pi (D -cQ)^{-1}  v   - \Pi  D^{-1} v  =\Pi (D -cQ)^{-1}  v  - \pi = {\bf 0}.
\end{align}
Similarly since $Q^T \pi = {\bf 0}$ we have $(D- c Q^T)\pi = D \pi$ and so $\pi  = (D - c Q^T)^{-1} D \pi = (D - cQ^T)^{-1} \Pi v $. Hence
\begin{align}
\label{Spart2}
M^T v =    (D -c Q^T)^{-1} \Pi v   -   D^{-1} \Pi v = \pi  -   \Pi D^{-1}  v = \pi  -  \Pi {\bf 1}   = \pi - \pi ={\bf 0}.
\end{align}
Combining \eqref{Spart1} and \eqref{Spart2} proves \eqref{Scompwisepositivevectorv} and shows that $(M + M^T)$ is a symmetric nonpositive definite matrix. Therefore \eqref{Ssuffcondtoshow3} holds and this finishes the proof of part (C) of this theorem.
\end{proof}

\section{Approximate formula for $m_{ \textnormal{eq} }(c)$} \label{Ssec:approxformula}

In this section we examine the approximate formula for $m_{ \textnormal{eq} }(c)$ and discuss why the approximation error is likely to be small.  Recall from Theorem \ref{Stheorem_general_supp} that $\lambda_1,\dots,\lambda_n$ are the eigenvalues of matrix $\tilde{Q}$. Among these $\lambda_1 = 0$ while the eigenvalues $\lambda_2,\dots,\lambda_n$ have negative real parts. Using part (B) of Theorem \ref{Stheorem_general_supp} gives us this exact formula for $m_{ \textnormal{eq} }(c)$:
\begin{align}
\label{Sresolventformula}
m_{ \textnormal{eq} }(c) = k_{ \textnormal{in} } \left[  \alpha_1 +   \sum_{i=2}^n  \left( \frac{ \alpha_i }{1 - c \lambda_i } \right)   \right].
\end{align}
From limits \eqref{Smdeterministic} and \eqref{Smstatic} we can conclude that
\begin{align}
\label{Salpharelations}
\alpha_1=  \frac{ m^{ \textnormal{(det)} }_{ \textnormal{eq} } }{ k_{ \textnormal{in} } }   \qquad \textnormal{and} \qquad \sum_{i=2}^n  \alpha_i = \left( \frac{m^{ \textnormal{(static)} }_{ \textnormal{eq} } - m^{ \textnormal{(det)} }_{ \textnormal{eq} } }{k_{ \textnormal{in} }}\right).
\end{align}
%
Let $\theta$ denote the following weighted combination of eigenvalues $\lambda_2,\dots,\lambda_n$
\begin{align}
\label{Sdefn_hatlambda}
\theta = -  \frac{  \sum_{i=2}^n \lambda_i \alpha_i }{\sum_{i=2}^n \alpha_i}.
\end{align}
Recall from the main paper that $\theta$ is the \emph{normalized} Dirichlet form which is always positive. We proposed the following approximate formula for $m_{ \textnormal{eq} }(c)$
\begin{align}
\label{Smceqapprox}
\hat{m}_{ \textnormal{eq} }(c) = m^{ \textnormal{(det)} }_{ \textnormal{eq} } +   \left(  \frac{m^{ \textnormal{(static)} }_{ \textnormal{eq} } - m^{ \textnormal{(det)} }_{ \textnormal{eq} }}{1 + c \theta}  \right).
\end{align}

We define the relative error between the exact value $m_{ \textnormal{eq} }(c)$ and its approximation $\hat{m}_{ \textnormal{eq} }(c)$ by
\begin{align*}
\mathcal{E}(c) =  \left( \frac{ m_{ \textnormal{eq} }(c) - \hat{m}_{ \textnormal{eq} }(c) }{ m^{ \textnormal{(static)} }_{ \textnormal{eq} } - m^{ \textnormal{(det)} }_{ \textnormal{eq} }} \right).
\end{align*}
Due to limits  \eqref{Smdeterministic} and \eqref{Smstatic}, we know that this error function is contained between $0$ and $1$. Our goal is to argue that this relative error function has a relatively small magnitude. For this purpose we define a change of variables as
\begin{align*}
x = \frac{c \theta}{1 + c \theta}.
\end{align*}
Note that as $c$ goes from $0$ to $\infty$, $x$ goes from $0$ to $1$. Define a function $f$ by
\begin{align*}
f(x) =   \sum_{i=2}^n  \frac{ \beta_i (1 - x) }{1 -   \left(1 +  \frac{\lambda_i}{\theta}  \right) x  }
\end{align*}
where
  \begin{align*}
\beta_ i = \frac{ \alpha_i }{\sum_{i=2}^n \alpha_i  } =  \alpha_i \left(  \frac{ k_{ \textnormal{in} } }{m^{ \textnormal{(static)} }_{ \textnormal{eq} } - m^{ \textnormal{(det)} }_{ \textnormal{eq} } }  \right)
\end{align*}
for each $i=2,\dots,n$. Replacing $c$ by $x/( (1-x) \theta )$ in \eqref{Sresolventformula} and using \eqref{Salpharelations} we get
\begin{align*}
m_{ \textnormal{eq} }(c) =  m^{ \textnormal{(det)} }_{ \textnormal{eq} } +   \left(  m^{ \textnormal{(static)} }_{ \textnormal{eq} } - m^{ \textnormal{(det)} }_{ \textnormal{eq} }  \right)f(x).
\end{align*}
Similarly $\hat{m}_{ \textnormal{eq} }(c) $ can be written as
\begin{align*}
\hat{m}_{ \textnormal{eq} }(c) = m^{ \textnormal{(det)} }_{ \textnormal{eq} } +   \left(  m^{ \textnormal{(static)} }_{ \textnormal{eq} } - m^{ \textnormal{(det)} }_{ \textnormal{eq} } \right)(1 - x)
\end{align*}
which allows us to express the relative error as
\begin{align}
\label{Ssimpl_eerfunction}
\mathcal{E}(c) = f(x) - 1 + x.
\end{align}

 Note that
\begin{align}
\label{Sbasicfrelns}
f(0) = \sum_{i=2}^n \beta_i = 1 \qquad  \textnormal{and} \qquad  \theta = -  \sum_{i=2}^n \lambda_i \beta_i.
\end{align}
We can compute the first and second order derivatives of function $f(x)$ as
\begin{align*}
f'(x) = \frac{1}{\theta} \sum_{i=2}^n  \frac{ \lambda_i \beta_i  }{ \left( 1 -   \left(1 +  \frac{\lambda_i}{\theta}  \right) x   \right)^2 } \qquad \textnormal{and}  \qquad
f''(x) =\frac{2}{\theta} \sum_{i=2}^n  \frac{ \lambda_i \beta_i \left(1 +  \frac{\lambda_i}{\theta}  \right) }{ \left( 1 -   \left(1 +  \frac{\lambda_i}{\theta}  \right) x   \right)^3 }.
\end{align*}
Therefore from \eqref{Sbasicfrelns} we obtain
\begin{align*}
f'(0) = \frac{1}{\theta} \sum_{i=2}^n   \lambda_i \beta_i = -1  \qquad \textnormal{and}  \qquad
\kappa:= \frac{f''(0)}{2} =\frac{1}{\theta} \sum_{i=2}^n   \lambda_i \beta_i \left(1 +  \frac{\lambda_i}{\theta}  \right) =  \sum_{i=2}^n \beta_i \left( 1 - \frac{ \lambda_i }{\theta} \right)^2.
\end{align*}
We can see that $\kappa = f''(0)/2$ measures the weighted ``spread" of the eigenvalues $\lambda_2,\dots,\lambda_n$ around $-\theta$, where the weights are given by $\beta_2,\dots,\beta_n$. Numerical experiments indicate that generally only a few of these weights are significant while the others are negligibly small. This is precisely the situation where this spread is small since $\theta = -  \sum_{i=2}^n \lambda_i \beta_i$. Hence we can safely assume that $\kappa$ is small. The function $f(x)$ is real-analytic at $x= 0$ and its Taylor series expansion is given by
\begin{align*}
f(x) &= f(0) + f'(0) x + \frac{f''(0)}{2}x^2 + \dots  = 1 - x + \frac{f''(0)}{2}x^2+\dots.
\end{align*}
Using \eqref{Ssimpl_eerfunction} we can conclude that the relative error behaves like $\kappa x^2$ which is small since $\kappa$ is small.

\section{Stochastic Focusing example}

Consider the  stochastic focusing network given in \cite{Paulsson00}. It involves three species: substrate {\bf S}, product {\bf P} and enzyme {\bf E}.  The molecules of substrate {\bf S} are produced constitutively at rate $k_{ \textnormal{in} }$ and converted into product {\bf P} with rate constant $k_p$. Both subtrate and product molecules degrade spontaneously at rates $k_a e$ and $\delta_p$ respectively, where $e$ denotes the current state or abundance level of enzyme {\bf E}. These reactions can be expressed as
\begin{align}
\label{Sdefn:reac_net_sf_supp}
\emptyset   \xrightleftharpoons[  k_a  e  ]{k_{ \textnormal{in} } } {\bf S}  \stackrel{k_p}{ \longrightarrow }   {\bf P}    \stackrel{ \delta_p }{ \longrightarrow } \emptyset.
\end{align}
Let $m^{({\bf S})}(t) $ and $m^{({\bf P})}(t) $ denote the expected abundance level at time $t$ of substrate and product molecules respectively. Furthermore assume that the limit
\begin{align}
\label{Sss_substrate_supp}
 m^{({\bf S})}_{ \textnormal{eq} }  = \lim_{ t \to \infty} m^{({\bf S})}(t)
\end{align}
exists. From the reaction network \eqref{Sdefn:reac_net_sf_supp} it is immediate that
\begin{align*}
\frac{d m^{({\bf P})}(t)   }{dt} = k_p m^{({\bf S})}(t) - \delta_p m^{({\bf P})}(t).
\end{align*}
Solving this ODE we get
\begin{align*}
m^{({\bf P})}(t)  = m^{({\bf P})}(0) e^{ -\delta_p t } + k_p   \int_0^t   e^{ -\delta_p (t-s) } m^{({\bf S})}(s)ds.
\end{align*}
Therefore using limit \eqref{Sss_substrate_supp} we can conclude that
\begin{align*}
 m^{({\bf P})}_{ \textnormal{eq} }  = \lim_{ t \to \infty} m^{({\bf P})}(t) = \lim_{t \to \infty} k_p   \int_0^t   e^{ -\delta_p (t-s) } m^{({\bf S})}(s)ds = \frac{k_p}{\delta_p} m^{({\bf S})}_{ \textnormal{eq} } .
\end{align*}

\bibliographystyle{acm}
\newcommand{\doi}[1]{\href{http://dx.doi.org/#1}{doi:#1}}
 \newcommand{\noop}[1]{}

\end{document}